\documentclass[12pt,onecolumn]{IEEEtran}
\pdfoutput=1
\newcommand{\shrinkpaper}{}

\usepackage{graphicx}

\usepackage{psfrag}
\usepackage[numbers,sort,compress]{natbib}      

\usepackage{algorithm} 
\usepackage{algpseudocode}

\usepackage{changebar}

\usepackage[version=0.96]{pgf}
\usepackage{tikz}
\usetikzlibrary{automata,petri}

\usepackage{color}
\usepackage{amscd}
\usepackage{amsfonts}

\usepackage{dsfont}

\usepackage{multirow}
\usepackage{stfloats}

\usepackage{ifthen}

\usepackage{amsmath} 
\usepackage{amssymb} 

\makeatletter
\let\theoremstyle\@undefined

\usepackage{amsthm} 

\newtheoremstyle{dcsctudelft}
  {8pt}
  {}
  {\itshape}
  {}
  {\bfseries}
  {.}
  {3pt}
  {}
  
  \theoremstyle{dcsctudelft}
  
\usepackage{macros-private}
\usepackage{variables-private}
\usepackage{figures-private}

\begin{document}

\title{Synchronization of a class of cyclic discrete-event systems describing legged locomotion}

        
\author{G.A.D.\ Lopes, B. Kersbergen, B.\ De\ Schutter, T.J.J.\ van\ den\ Boom,  and R.\ Babu\v{s}ka
        
\thanks{Corresponding author G.A.D.~Lopes. Tel. +31-15-2785489. All authors are with the Delft Center for Systems and Control, Delft University of Technology, The Netherlands,
e-mail: \{g.a.delgadolopes, b.kersbergen, a.j.j.vandenboom, b.deschutter, r.babuska\}@tudelft.nl}
\thanks{Manuscript received xxxxxx, 20xx; revised xxxxx, 20xx.}}

\markboth{}
{Lopes \MakeLowercase{\textit{et al.}}: }

\maketitle

%
%

\begin{abstract}
It has been shown that max-plus linear systems are well suited for applications in synchronization and scheduling, such as the generation of train timetables, manufacturing, or traffic. In this paper we show that the same is true for multi-legged locomotion. In this framework, the max-plus eigenvalue of the system matrix represents the total cycle time, whereas the max-plus eigenvector dictates the steady-state behavior. Uniqueness of the eigenstructure also indicates uniqueness of the resulting behavior. For the particular case of legged locomotion, the movement of each leg is abstracted to two-state circuits: swing and stance (leg in flight and on the ground, respectively). The generation of a gait (a manner of walking) for a multiple legged robot is then achieved by synchronizing the multiple discrete-event cycles via the max-plus framework. By construction, different gaits and gait parameters can be safely interleaved by using different system matrices. In this paper we address both the transient and steady-state behavior for a class of gaits by presenting closed-form expressions for the max-plus eigenvalue and max-plus eigenvector of the system matrix and the coupling time. The significance of this result is in showing guaranteed robustness to perturbations and gait switching, and also a systematic methodology for synthesizing controllers that allow for legged robots to change rhythms fast.
\end{abstract}


\IEEEpeerreviewmaketitle

%
%
%
%
%
%
%
%

%
%

\section{Introduction}

Synchronization of cyclic processes is important in many fields, including manufacturing \cite{Zhou92}, transportation \cite{heidergott01}, genomics \cite{Shedden02}, and neuroscience \cite{yamaguchi03,holmes06}, etc (see references within \cite{Dorfler12}). In this paper we focus on a class of multiple concurrent two-state cyclic systems with a direct application to legged locomotion. Our motivation is the requirement of legged mobile robots to cope with unstructured terrains in a flexible way by smoothly and effectively switching between different gaits.

Legged systems are traditionally modeled using cross-products of circles in the phase space of the set of continuous time gaits. Holmes et al.~\cite{holmes06} give an extensive review of dynamic legged locomotion. The \CPG~\cite{Ijspeert08} approach to design motion controllers lies in assembling sets of error functions to be minimized that cross-relate the phases of multiple legs, resulting in attractive limit cycles for the desired gait. Switching gaits online is typically not addressed since in the \CPG~framework switching must be modeled as a hybrid system. Additionally, implementing ``hard constraints'' on the configuration space \cite{Haynes09} can be quite complex mainly due to the combinatorial nature of the gait space, and often it comes at the cost of dramatically increasing the complexity of the controller.

As an alternative to the common continuous time modeling approach, we introduce an abstraction to represent the combinatorial nature of the gait space for multi-legged robots into ordered sets of leg index numbers. This abstraction combined with max-plus linear equations allows for systematic synthesis and implementation of motion controllers for multi-legged robots where gait switching is natural and the translation to continuous-time motion controllers is straightforward \cite{lopes09}. The methodology presented is particularly relevant for robots with four, six, or higher numbers of legs where the combinatorial nature of the gait space starts to play an important role. For a large number of legs it is not obvious in which order each leg should be in swing or in stance. Most legged animals, in particular large mammals, are known to walk and run with various gaits on a daily basis, depending on the terrain or how fast they need to move. The discrete-event framework presented in this paper enables the same behavior for multi-legged robots. Mathematical properties are derived for this framework, giving extra insight into the resulting robot motion.

The main mathematical representation employed in this paper are switching max-plus linear equations. By defining the state variables to represent the time at which events occur, systems of linear equations in the \emph{max-plus algebra} \cite{cuninghame79,baccelli92,heidergott06} can model a class of timed discrete-event systems. Specifically, max-plus liner systems are equivalent to (timed) Petri nets \cite{Peterson81} where all places have a single incoming and a single outgoing arc. Max-plus linear systems inherit a large set of analysis and control synthesis tools thanks to many parallels between the max-plus-linear systems theory and the traditional linear systems theory.
Discrete-event systems that enforce synchronization can be modeled in this framework. Max-plus algebra has been successfully applied to railroads \cite{braker91,heidergott01}, queuing systems \cite{heidergott00}, resource allocation \cite{gaubert98}, and recently to image processing \cite{bede09} and legged locomotion \cite{lopes09,lopes10}.

The contributions of this paper are the following: we present a class of max-plus linear systems that realize the synchronization of the legs of a robot. 
Next, we derive closed-form expressions for the max-plus eigenvalue and eigenvector of the system matrix, and show that the max-plus eigen-parameters are max-plus unique, implying a unique steady-state behavior. 
 This result is then used to compute the coupling time, which characterizes the transient behavior. 
The importance of having closed-form expressions and uniqueness of the max-plus eigen-structure is that, not only can one compute these parameters very fast without recurring to simulations or numerical algorithms (e.g. Karp's algorithm \cite{baccelli92}), but one has also guarantees of uniqueness: the motion of the robot will always converge in a finite number of steps to the same prescribed behavior, regardless of gait changes or disturbances. This reassurance is fundamental when designing gait controllers for robotics. Additionally we present a low least number of steps needed to reach steady-state motion after changing gaits. This paper is focused on the general mathematical properties of a class of discrete-event systems that describe legged locomotion, and not on the actual implementation of gait controllers for robots, as presented previously in \cite{lopes09}.

In \refsec{mll} we revisit the fundamentals of legged locomotion with special emphasis on gait generation and show that max-plus algebra can be used in the modeling of the synchronization of multiple legs. In \refsec{maxplus} we briefly review relevant concepts from the theory of max-plus algebra 
and in \refsec{combinatorial} we present a class of parameterizations for the gait space. Given such class, we derive a number of properties, such as the max-plus eigen-structure of the system matrices (\refsec{eigen}), their graph representation (\refsec{graph}), and the coupling time (\refsec{coupling}).
\reffig{layout} illustrates the structure of the contributions presented in this paper. 

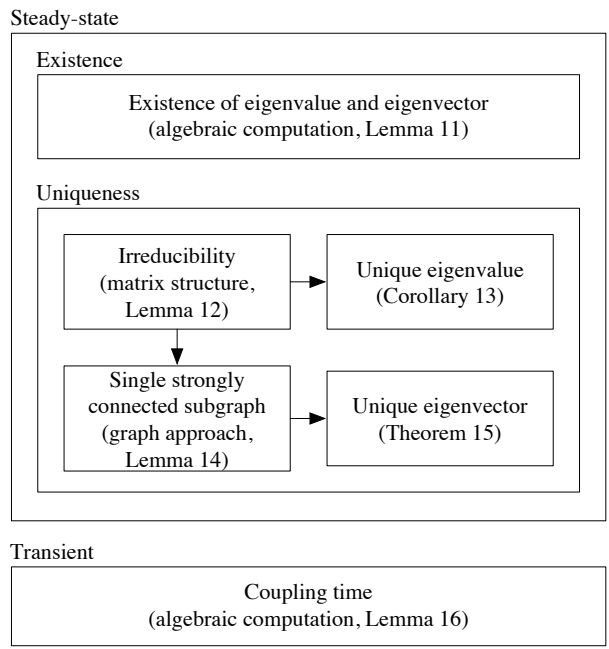
\figura{layout}{8cm}{Structure of the contributions of this paper. We analyze both the steady-state and the transient behavior of a class of cyclic discrete-event systems that are well suited to model legged locomotion.}{}
%
%

\section{Modeling legged locomotion}\labelsec{mll}

In literature \cite{raibert86,raibert89,full99,holmes06,grillner11}, the study of legged locomotion is approached from two main directions: the signal generation side, where emphasis is placed on the classes of signals that result in periodic locomotion behavior independent of the physical platform; and the mechanics side, where the (hybrid) Newtonian mechanics models are analyzed independently of the driving control signal. We focus on the first approach, by restating the traditional view of periodic gaits for legged systems being defined in the $n$-torus: Cartesian products of circles each representing an abstract phase that parameterizes the position of each leg in the Euclidean space. 
Such an abstraction serves as a platform for the models of ``networks of phase oscillators'' and \CPGs, introduced in the earlier works of Grillner \cite{grillner85} and Cohen et al. \cite{cohen88}. These are now accepted by both biology and robotics communities as standard modeling tools \cite{holmes06}. 
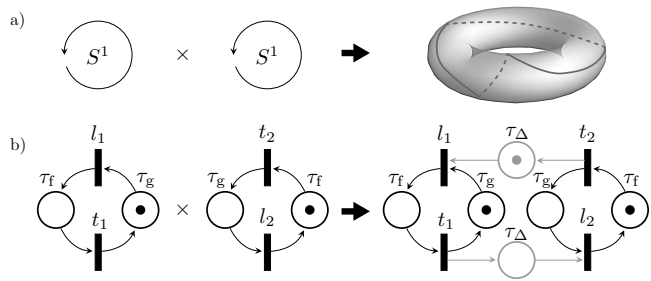
\figura{abstractions}{8.5cm}{Modeling of legged locomotion. a) The configuration space of two oscillators is a torus. Synchronization is achieved by constructing attractive limit cycles, represented by the curve on the torus. b) Discrete-event representations of multiple circuits can be modeled as concurrent cyclic Petri nets. Synchronization is achieved by adding extra places, represented by $s_{2}$ and $s_{2}$.}{}

In this paper we abstract beyond the notion of a continuous phase to consider concurrent circuits of discrete events. Taking a Petri net modeling approach, during ground locomotion the places represent leg stance (when the foot is touching the ground and supporting the body) and leg swing (when the foot and all parts of the leg are in the air). The transitions represent leg touchdown and lift off. This labeling is most convenient for ground locomotion, but one should be aware that the framework presented in this paper is valid for other types of locomotion where the phases of the various limbs need to be synchronized, such as swimming or flight. \reffig{abstractions} illustrates the conceptual difference between our approach (\reffig{abstractions}.b) where each phase is represented by a circuit in the discrete-event systems domain and the traditional continuous phase central pattern generator (\reffig{abstractions}.a). Here, phase synchronization is enforced on the torus by implementing controllers that achieve stable limit cycles \cite{klavins02}. In this paper we translate timed event graphs\footnote{We restrict ourselves to a class of timed Petri nets called \emph{timed event graphs} such that a one-to-one translation to max-plus linear systems is possible (see \cite{heidergott06}, chapter 7). In timed event graphs each place can have one single incoming arc and one single outgoing arc.} into the equivalent representation as max-plus linear systems, and achieve synchronization by designing the system matrices appropriately.
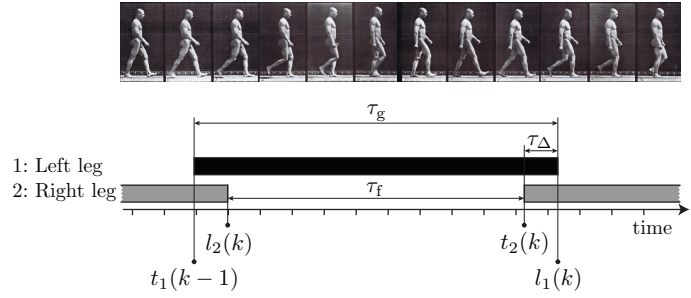
\figura{walking-pattern-compressed}{9cm}{Illustration of a walking pattern, photos by Muybridge \cite{Muybridge01}. The solid bars, following Hildebrand's diagram notation \cite{hildebrand65}, indicate that the leg is in stance (foot touching the ground) for $\Tg$ time units, and white space that the leg is in swing (foot in flight) for $\Tf$ time units. The time length when both feet are touching the ground is called the double-stance time $\Td$. We use the notation $t_1(k-1)$ to represent the time instant when leg 1 (left leg) touches the ground, and $l_1(k)$ when it initiates a swing. The parameter $k$ is the ``step counter''.}{}
\!\!Consider the leg synchronization of a biped robot. In this case, only two legs need to be synchronized during locomotion. For a typical walking motion (no aerial phase) the left leg should only lift off the ground after the right leg has touched down, to make sure the robot does not fall due to lack of support. This simple synchronization requirement can be captured by introducing state variables for the transition events defined as follows: let $\tL{i}(k)$ be the
time instant leg $i$ lifts off the ground and $\tT{i}(k)$ be the
time instant it touches the ground, both for $k$-th iteration, where $k$ is considered to be a global event or ``step'' counter.  
Enforcing that the time instant when the leg touches the ground must equal the time
instant it lifted off the ground for the last time plus the time it
is in swing (denoted $\Tf$) is realized by:
\NEn{\tT{i}(k)=\tL{i}(k)+\Tf\labeleq{touch}.}
A similar relation can be derived for the lift off time:
\NEn{\tL{i}(k)=\tT{i}(k-1)+\Tg,\labeleq{lift}}
where $\Tg$ is the stance time and $\tT{i}(k-1)$ refers to the previous
iteration such that equations \refeq{touch} and \refeq{lift} can be
used iteratively. For this system we have that $\Tf>0$ and also $\Tg>0$.
Synchronization of the cycles of two legs can be achieved by introducing a double stance time parameter, denoted $\Td$, representing that after each leg touchdown both legs must stay in stance for at least $\Td$ time units (see \reffig{walking-pattern-compressed}). This is captured by the following equations:
\NEn{\tL{1}(k) \ele \max\left(\tT{1}(k-1)+\Tg,\tT{2}(k-1)+\Td\right)\labeleq{liftsync}\\
\tL{2}(k) \ele \max\left(\tT{2}(k-1)+\Tg,\tT{1}(k)+\Td\right)\labeleq{liftsyncb}}
Equation \refeq{liftsync} enforces simultaneously that
\mbox{leg $1$} stays at least $\Tg$ time units in stance and will only
lift off at least $\Td$ time units after leg $2$ has touched down. When
both conditions are satisfied, lift off takes place. Equation \refeq{liftsyncb} is analogous for leg $2$.
Note that the parameters $\Tf, \Tg$, and $\Td$ are in fact the \emph{minimal} swing, stance, and double stance times, respectively, as opposed to the exact times. Note additionally that these equations are non-linear but they rely only on the ``max'' and ``plus'' operations. This motivates the use of the theory of the max-plus algebra to find parsimonious discrete-event models for legged locomotion. Although the previous set of equations were designed for a walking behavior with no aerial phases, in practice they can also be used for running by choosing the double stance time to be negative, i.e. by enforcing that one leg lifts off $\Td$ time units before the other leg touches the ground. However, the results presented in this paper focus on the case when $\Td\geq 0$ such that stability (in the sense of ensuring a desired minimum number of legs on the ground simultaneously) is guaranteed.

%

\section{Max-plus algebra}
\labelsec{maxplus}

In the early sixties the fact that certain classes of discrete-event
systems can be described by models using the operations $\max$ and $+$
has been discovered independently by a number of researchers, among
whom Cuninghame-Green \cite{Cun:60,Cuninghame62} and Giffler
\cite{Giffler60,Gif:63,Gif:68}.  These discrete-event systems are
called max-plus-linear systems since the model that describes their
behavior becomes ``linear'' when formulated it in the max-plus
algebra \cite{baccelli92,cuninghame79,heidergott06}, which has
maximization and addition as its basic operations.  More specifically,
discrete-event systems in which only synchronization and no concurrency
or choice occur can be modeled using the operations maximization
(corresponding to synchronization: a new operation starts as soon as
all preceding operations have been finished) and addition
(corresponding to the duration of activities: the finishing time of an
operation equals the starting time plus the duration).  Some examples
of max-plus linear discrete-event systems are production systems,
railroad networks, urban traffic networks, queuing systems, and array
processors \cite{baccelli92,cuninghame79,heidergott06}.

An account of the pioneering work of Cuninghame-Green on
max-plus system theory has been given in
\cite{cuninghame79}.  Related work has been done by Gondran and Minoux
\cite{GonMin:76,GonMin:84,GonMin:87}.  In the eighties the topic attracted new interest due to the research of Cohen, Dubois,
Moller, Quadrat, Viot \cite{CohDub:83,cohen85,CohMol:89}, Olsder
\cite{Ols:86,OlsRoo:88,OlsRes:90}, and Gaubert
\cite{Gau:90,Gau:92,Gau:93}, which resulted in the publication of
\cite{baccelli92}.  Since then, several other researchers have entered
the field.  For an historical overview we refer the interested reader
to \cite{gaubert97,heidergott06,DeSvan:08-007}.
\Gshrink
In this section we give an introduction to the max-plus algebra.  This
section is based on \cite{baccelli92,cuninghame79}, where a complete overview
of the max-plus algebra can be found.
\Gshrink
The basic operations of the max-plus algebra
are maximization and addition, which will be
represented by $\oplus$ and $\otimes$ respectively
\ZNE{
   x \oplus y = \max(x,y) \quad \text{ and } \quad
   x \otimes y = x + y, 
}
for $x,y \in \Rmax \bydefinition \Reals \cup \{ -\infty \}$. 
The zero
element for $\oplus$ in $\Rmax$ is $\ep \bydefinition -\infty$ and the
unit element for $\otimes$ is $\ze \bydefinition 0$.  
\Gshrink
The structure $(\Rmax,\oplus,\otimes)$ is called
the max-plus algebra \cite{baccelli92,cuninghame79}. The operations
$\oplus$ and $\otimes$ are called the max-plus-algebraic
addition and max-plus-algebraic multiplication
respectively since many properties and concepts
from linear algebra can be translated to
the max-plus algebra by replacing
$+$ by $\oplus$ and $\times$ by $\otimes$.
\Gshrink
The max-plus algebra is a typical example of a class of algebraic
structures called commutative dioids. In a dioid the additive
operation $\oplus$ is associative, commutative, and idempotent, and it
has a zero element; the multiplicative operation $\otimes$ is
commutative and associative and it has an identity element; the
additive zero element is absorbing for $\otimes$; and $\otimes$ is
left and right distributive w.r.t.~$\oplus$.
\Gshrink
Let $r \in \Reals$. The $r$th max-plus power of $x \in \Reals$ is
denoted by $\Mpower{x}{r}$ and corresponds to $r x$ in conventional
algebra. If $x \in \Reals$ then $\Mpower{x}{0} = 0$ and the inverse
element of $x$ w.r.t.~$\otimes$ is $\Mpower{x}{-1} = -x$.  There is no
inverse element for $\ep$ since $\ep$ is absorbing for $\otimes$.  If
$r > 0$ then $\Mpower{\ep}{r} = \ep$.  If $r < 0$ then
$\Mpower{\ep}{r}$ is not defined.  In this paper we have
$\Mpower{\ep}{0} = 0$ by definition.
The rules for the order of evaluation of the max-plus operators are
similar to those of conventional algebra. So max-plus power
has the highest priority, and max-plus multiplication
$\otimes$ has a higher priority than max-plus addition
$\oplus$.
\Gshrink
Throughout this paper the $i,j$ element of a matrix $\A$ is denoted by
$[\A]_{ij}$.  The matrix $\Mz_{m \times n}$ is the $m$ by $n$
max-plus zero matrix: $[ \Mz_{m \times n} ]_{i j} = \ep$
for all $i,j$.  The matrix $\Mid_n$ is the $n$ by $n$
max-plus identity matrix: $[\Mid_n]_{i i} = \ze$ for all $i$
and $[\Mid_n]_{i j} = \ep$ for all $i,j$ with $i \neq j$. We also
define the $m$ by $n$ max-plus ``one'' matrix $\Mone_{m \times n}$ such
that $[\Mone]_{i j}=\ze=0$ for all $i,j$. If the dimensions of $\Mz$,
$\Mid$, $\Mone$ are omitted in this paper, they should be clear from
the context.
\Gshrink
The basic max-plus-algebraic operations are extended to matrices as
follows.  If $A, B \in \Rmax^{m \times n}$, $C \in \Rmax^{n \times p}$
then
\begin{align}
   [ A \oplus B ]_{i j} &= [A]_{i j} \oplus [B]_{i j} = \max([A]_{i j}
   ,
   [B]_{i j} ) \\
   [ A \otimes C ]_{i j} &= \bigoplus_{p=1}^{n} [A]_{i p} \otimes
   [C]_{p j} = \max_{p=1,\dots,n} ( [A]_{i p} + [C]_{p j} )
\end{align}
for all $i,j$.  Note the analogy with the definitions of matrix sum
and product in conventional linear algebra.  The max-plus product of
the scalar $\alpha \in \Rmax$ and the matrix $A \in \Rmax^{m \times
   n}$ is defined by $[ \alpha \otimes A]_{i j} = \alpha \otimes
[A]_{i j}$ for all $i,j$.  The max-plus matrix power of $A \in
\Rmax^{n \times n}$ is defined as follows: $\Mpower{A}{0} = \Mid_n$
and $\Mpower{A}{p} = A \otimes \Mpower{A}{p-1}$ for $p\geq 1$.
\Gshrink
\begin{theorem}[see \cite{baccelli92}, Th 3.17]
Consider the following system of linear equations in the max-plus
algebra:
\NEn{
   x = A \otimes x \oplus b
  \label{eq_max_plus_linear}
}
with $A \in \Rmax^{n \times n}$ and $b, x \in \Rmax^{n \times 1}$.
Now let
\NEn{
   \As \bydefinition \bigoplus^{\infty}_{p=0}\Mpower{\A}{p} \enspace.
}
If $\As$ exists then 
\NEn{\x=\As\mtimes \bb\labeleq{explicitsolution}} 
solves the system of max-plus
linear equations \eqref{eq_max_plus_linear}
\end{theorem}
\begin{definition}
   The matrix $A\in\Rmax^{n\times n}$ is called nilpotent if there
   exists a finite positive integer $p_0$ such that for all integers $p \geq p_0$
   we have $\Mpower{A}{p} = \Mz$.
\end{definition}
It is easy to verify that if $A\in\Rmax^{n\times n}$ is nilpotent then
$p_0 \leq n$.
\Gshrink
For $\A,\B\in\Rmax^{n\times m}$ we say that $A$ \emph{overcomes} $B$,
written as $A\geq B$ if $A\mplus B=A$ (i.e., $[A]_{i j}\geq [B]_{i j}$
for all $i,j$).
\Gshrink
A directed graph $\mathcal{G}$ is defined as an ordered pair
($\mathcal{V}$,$\mathcal{A}$), where $\mathcal{V}$ is a set of
vertices and $\mathcal{A}$ is a set of ordered pairs of vertices.  The
elements of $\mathcal{A}$ are called arcs.  A loop is an arc of the
form $(v,v)$.
\Gshrink
Let $\mathcal{G}=(\mathcal{V},\mathcal{A})$ be a directed graph with
$\mathcal{V}=\{v_1,v_2, \ldots, v_n\}$.  A path $p$ of length $l$ is a
sequence of vertices $v_{i_1}$, $v_{i_2}$, \ldots, $v_{i_{l+1}}$ such
that ($v_{i_k}$, $v_{i_{k+1}}$) $\in \mathcal{A}$ for
$k=1,2,\ldots,l$.  We represent this path by $v_{i_1} \rightarrow
v_{i_2} \rightarrow \ldots \rightarrow v_{i_{l+1}}$ and we denote the
length of the path by $|p|_{\mathrm{l}}=l$. Vertex $v_{i_1}$ is the
initial vertex of the path and $v_{i_{l+1}}$ is the final vertex of
the path. The set of all paths of length $l$ from vertex $v_{i_1}$ to
$v_{i_l}$ is denoted by $\path(v_{i_1},v_{i_l};l)$.
\Gshrink
When the initial and the final vertex of a path coincide, we have a
circuit. An elementary circuit is a circuit in which no vertex appears
more than once, except for the initial vertex, which appears exactly
twice.  A directed graph $\mathcal{G}=(\mathcal{V},\mathcal{A})$ is
called \emph{strongly connected} if for any two different vertices
$v_i$, $v_j \in \mathcal{V}$ there exists a path from $v_i$ to $v_j$.
\Gshrink
If we have a directed graph $\mathcal{G}=(\mathcal{V},\mathcal{A})$
with $\mathcal{V}= \set{1,2}{n}$ and if we associate a real number
$[A]_{i j}$ with each arc $(j,i) \in \mathcal{A}$, then we say that
$\mathcal{G}$ is a weighted directed graph. We call $[A]_{i j}$ the
weight of the arc $(j,i)$.  Note that the first subscript of $[A]_{i j}$
corresponds to the final (and not the initial) vertex of the arc
$(j,i)$. 
\begin{definition}[Precedence graph]
   Consider $A \in \Rmax^{n \times n}$.  The precedence graph of $A$,
   denoted by $\mathcal{G}(A)$, is a weighted directed graph with
   vertices $1$, $2$, \ldots, $n$ and an arc $(j,i)$ with weight $[A]_{i
      j}$ for each $[A]_{i j} \neq \ep 
      $.
\end{definition}
Let $A\in \Rmax^{n \times n}$ and consider $\mathcal{G}(A)$.
The weight $|p|_{\mathrm{w}}$
of a path $p: i_1 \rightarrow i_2 \rightarrow \ldots \rightarrow
i_{l+1}$ is defined as the sum of the weights of the arcs that compose
the path: $|p|_{\mathrm{w}} = [A]_{i_2 i_1} + [A]_{i_3 i_2} + \ldots +
[A]_{i_{l+1} i_{l}} = \bigotimes_{k=1}^{l}[A]_{i_{k+1}i_{k}}$.  The
average weight of a circuit is defined as the weight of the circuit
divided by the length of the circuit: $|p|_{\mathrm{w}}/|p|_{\mathrm{l}}$.
\begin{definition}[Irreducibility]
   A matrix $A \in \Rmax^{n \times n}$ is called irreducible
   if its precedence graph is strongly connected.
\end{definition}
\begin{definition}[Max-plus eigenvalue and
      eigenvector]
   Let $A \in \Rmax^{n \times n}$.  If there exist a number $\lambda
   \in \Rmax$ and a vector $v \in \Rmax^{n}$ with $v \neq \Mz_{n
      \times 1}$ such that $A  
  \otimes  
      v = \lambda 
      \otimes   
      v$, then we say that
   $\lambda$ is a max-plus eigenvalue of $A$ and that $v$ is
   a corresponding max-plus eigenvector of $A$.
\end{definition}
It can be shown that every square matrix with entries in $\Rmax$ has
at least one max-plus eigenvalue (see e.g.~\cite{baccelli92}).
However, in contrast to linear algebra, the number of max-plus
eigenvalues of an $n$ by $n$ matrix is in general less than $n$.  If a
matrix is irreducible, it has only one max-plus eigenvalue (see
e.g.~\cite{cohen85}). Moreover, if $v$ is a max-plus eigenvector of
$A$, then $\alpha \otimes v$ with $\alpha \in \Reals$ is also a
max-plus eigenvector of $A$.
\Gshrink
The max-plus eigenvalue has the following graph-theoretic
interpretation.  Consider $A \in \Rmax^{n \times n}$.  If
$\lambda_\mathrm{max}$ is the maximal average weight over all
elementary circuits of $
\mathcal{G}(A)  
$, then $\lambda_\mathrm{max}$ is a
max-plus eigenvalue of $A$.  For formulas and algorithms to determine
max-plus eigenvalues and eigenvectors the interested reader is
referred to~\cite{baccelli92,BraOls:93,cohen85,Kar:78} and the
references cited therein. Every circuit of $\mathcal{G}(A)$ with an
average weight that is equal to $\lambda_\mathrm{max}$ is called a
critical circuit.  The \emph{critical graph} $\graphc(A)$ of the
matrix $A$ is the set of all critical circuits. Let $\Integers$ be the set of positive, non-zero integers.
\begin{theorem}
  \labelth{coupling}
   Let $A$ be an irreducible matrix. Then there exists $c\in\Integers$ (the cyclicity of $A$), $\lambda\in\Reals$ (the unique max-plus eigenvalue of $A$), and $\couplingz\in\Integers$ (the coupling time of $A$) such that
   \begin{gather}
      \forall p\! \geq\! k_0\,:\:
      \Mpower{A}{p+c} = \Mpower{\lambda}{c} 
      \!\otimes\!  
      \Mpower{A}{p}
   \end{gather}
\end{theorem}
\begin{proof}
   See e.g.~\cite{baccelli92,cohen85,Gau:94}.
\end{proof}

%
%

\section{Legged locomotion via max-plus modeling}\labelsec{combinatorial}
 
In \refsec{mll} equations \refeq{touch}, \refeq{liftsync}, and \refeq{liftsyncb}  describe the synchronization constraints between two legs. We can generalize these equations by defining the following vectors for an $n$-legged robot:
\NEn{\tT{}(k)&=&\left[\tT{1}(k)~\cdots~\tT{n}(k)\right]^{T}\\
\tL{}(k)&=&\left[\tL{1}(k)~\cdots~\tL{n}(k)\right]^{T}}
Equation \refeq{touch} is then written as:
\NEn{\tT{}(k)&=&\Tf\mtimes\tL{}(k)\labeleq{xxa}
}
If one assumes that the synchronization is always enforced on the lift off time of a leg, equations \refeq{liftsync} and \refeq{liftsyncb} are written jointly as:
\NEn{\tL{}(k)&=&\Tg\mtimes\tT{}(k-1)\mplus \MTD\mtimes \tT{}(k)\mplus \MTDa \mtimes \tT{}(k-1),\labeleq{xxb}
}
where the matrices $\MTD$ and $\MTDa$ encode the synchronization between lift off of a leg related to a touchdown of the current event (as in equation \refeq{liftsyncb}) and a touchdown of the previous event (as in equation \refeq{liftsync}), respectively. The rationale behind this particular model is to prevent that a legged platform has too many legs in swing while walking\footnote{As mentioned previously in this paper we don't consider running, although it can still be achieved using the same class of models.} risking falling down. Synchronization constraints are always imposed on legs that are in stance and are about to enter swing: some legs should only swing if others are in stance (equation \refeq{xxb}). Once in swing, legs are never constrained to go into stance (equation \refeq{xxa}).

 Equations \refeq{xxa} and \refeq{xxb} are written in state-space form as:
\shrinkornot
{
\NEn{
\matris{c}{\tT{}(k)\\\tL{}(k)}\ele
\matris{c|c}{\Mz &\Tf \mtimes \Mid\\\hline \MTD& \Mz  }\mtimes
\matris{c}{\tT{}(k)\\\tL{}(k)}\nonumber\\&&\mplus
\matris{c|c}{\Mid &\Mz\\\hline \Tg \mtimes \Mid\mplus\MTDa& \Mid  }\mtimes
\matris{c}{\tT{}(k-1)\\\tL{}(k-1)}\labeleq{sssync}}
}
{
\NEn{
\matris{c}{\tT{}(k)\\\tL{}(k)}\ele
\matris{c|c}{\Mz &\Tf \mtimes \Mid\\\hline \MTD& \Mz  }\mtimes
\matris{c}{\tT{}(k)\\\tL{}(k)}\mplus
\matris{c|c}{\Mid &\Mz\\\hline \Tg \mtimes \Mid\mplus\MTDa& \Mid  }\mtimes
\matris{c}{\tT{}(k-1)\\\tL{}(k-1)}\labeleq{sssync}}
}
Define the matrices
\NEn{\Azero = \matris{c|c}{\Mz &\Tf \mtimes \Mid\\\hline \MTD& \Mz  };~~~~~~
\Aone=\matris{c|c}{\Mid &\Mz\\\hline \Tg \mtimes \Mid\mplus\MTDa& \Mid  }
}
Consider the full state $\x$ defined as
\ZNE{\x(k)=[\tT{}^{T}(k)~\tL{}^{T}(k)]^T.}
Equation \refeq{sssync} can then be written in simplified notation:
 \NEn{\x(k)= \Azero\mtimes\x(k)\mplus \Aone\mtimes\x(k-1).\labeleq{implicit}}
Note that additional max-plus identity matrices $\Mid$ are introduced in the diagonal of matrix $\Aone$. This results in the extra trivial constraints $\tT{i}(k+1)\geq \tT{i}(k)$ and $\tL{i}(k+1)\geq \tL{i}(k)$, also resulting in the final system matrix (defined in page \pageref{eq.Adefinition}, equation \refeq{Adefinition}) being irreducible. This is observed later on in \reflem{irreducibility}.
%

%
%

\subsection{A gait parameterization}

Consider a general legged robot where a two-event circuit is associated to each leg. We present a parsimonious representation of a walking gait of a robot by grouping sets of legs and specifying in what order they are allowed to cycle. 
\begin{definition}
Let $n$ be the number of legs in the robot and define $m$ as a number of leg groups. Let $\leg{1},\dots,\leg{\m}$ be ordered sets of integers such that
\NEn{
&&\bigcup_{p=1}^\m\leg{p}=\{1,\dots,n\},~
\forall i\neq j, \leg{i}\cap\leg{j}=\emptyset,\text{~and~}\forall i, \leg{i} \neq\emptyset}
i.e., the sets $\leg{p}$ form a partition of $\{1,\dots,n\}$. 
A gait $\gait$ is defined as an ordering relation of groups of legs:
\NEn{\gait = \leg{1}\mporder\leg{2}\mporder\cdots\mporder\leg{\m}\labeleq{ordering}}
The \emph{gait space} is the set of all gaits that satisfy the previous definitions. 
\end{definition}
By considering that each $\leg{p}$ contains the indices of a set of legs that are synchronized in phase, the previous ordering relation is interpreted in the following manner: the set of legs indexed by $\leg{i}$ swings synchronously. As soon as all legs in $\leg{i}$ touchdown then all legs in $\leg{i+1}$ initiate their swing motion. The same is true for $\leg{m}$ and $\leg{1}$, closing the cycle.
For example, a trotting gait, where diagonal pairs of legs move synchronously, for a quadruped robot as illustrated in \reffig{zebro}, is represented by:
\NEn{\gait_{\mathrm{trot}}=\{1,4\}\mporder\{2,3\}\labeleq{trottinggait}}
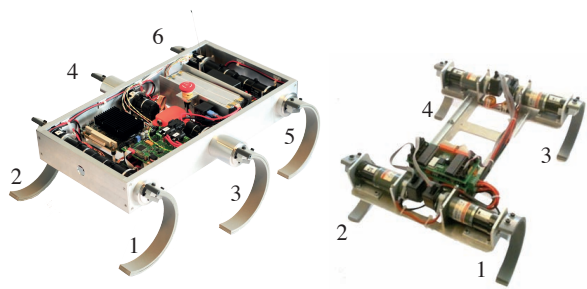
\figura{zebro}{8.5cm}{Walking robots with recirculating legs inspired by RHex \cite{saranli01}. Zebro robot on the left and RQuad on the right both developed at DCSC, Delft University of Technology. The numbers represent the leg index numbering assumed in this paper.}{}
The gait space defined above can represent gaits for which all legs have the same cycle time. As such, gaits where one leg cycles twice while another cycles only once are not captured by this model. Examples of such gaits are not common, but have been used on hexapod robots to transverse very inclined slopes sideways \cite{weingarten03}.
With the previous notation we can now derive the matrices $\MTD$ and $\MTDa$ in equation \refeq{sssync}:
\NEn{
\left[\MTD\right]_{pq}\ele\lista{ll}{\!\Td&\forall j \in \{ 1,\dots,\m\!-\!1\}; \forall p\in \leg{j+1}; \forall q\in \leg{j}
\\\!\zz&\text{otherwise}}
\labeleq{Pdefinition}\\
\left[\MTDa\right]_{pq}\ele\lista{ll}{\!\Td&\forall p\in \leg{1}; \forall q\in \leg{m}
\\\!\zz&\text{otherwise}}\labeleq{Qdefinition}}
\examples{For the trotting gait $\gait_{\mathrm{trot}}$ we obtain:
\NEn{\MTD_{\mathrm{trot}}=\matris{cccc}{
\zz & \zz &\zz &\zz\\
\Td & \zz &\zz &\Td\\
\Td & \zz &\zz &\Td\\
\zz & \zz &\zz &\zz}\text{~~and~~}\MTDa_{\mathrm{trot}}=
\matris{cccc}{
\zz & \Td &\Td &\zz\\
\zz & \zz &\zz &\zz\\
\zz & \zz &\zz &\zz\\
\zz & \Td &\Td &\zz}
}
}
Define the function $\flat$ that transforms a gait into a vector of integers:
\NEn{
\flat&:&\{[\leg{1}]_{1},\dots,[\leg{1}]_{i_{1}}\}\mporder\cdots\mporder\{[\leg{m}]_{1},\dots,[\leg{m}]_{i_{m}}\}\mapsto\\\nonumber
&&\!\!\!\left[[\leg{1}]_{1},\dots,[\leg{1}]_{i_{1}}\dots[\leg{m}]_{1},\dots,[\leg{m}]_{i_{m}}\right]^{T}
}
Using again the previous trotting example we get that  $\flat\left(\gait_{\mathrm{trot}}\right)=[1~4~2~3]^{T}$ (the symbol \emph{flat} ``$\flat$'' is chosen since it ``flattens'' the ordered collection of ordered sets of a gait into a vector).
Note that the gaits $\{1,4\}\mporder\{2,3\}$ and $\{4,1\}\mporder\{2,3\}$ although resulting in indistinguishable motion in practice, have different mathematical representations since $\flat\left(\{1,4\}\mporder\{2,3\}\right)\neq\flat\left(\{4,1\}\mporder\{2,3\}\right)$.
\begin{definition}
A gait $\gaitn$ is called a \emph{normal gait} if the elements of the vector $\flat\left(\gaitn\right)$ are sorted increasingly.
\end{definition}
For a gait $\gait$, define the similarity matrix $\CCb\in \Rmax^{n \times n}$ as:
\NEn{
\left[\CCb\right]_{ij}=\lista{ll}{\ze&\text{~if~}\left[\flat(\gait)\right]_{i}=j\\
\ep&\text{~otherwise}}, \forall i,j\in\{1,\dots,n\}
}
The similarity matrix $\CCb$ is such that
\ZNE{\CCb\mtimes\CCb^{T}=\CCb^{T}\mtimes\CCb=\Mid.
}
\examples{
The similarity matrix associated with the trotting gait $\gait_{\mathrm{trot}}$ is:
\NEn{\CCb_{\mathrm{trot}}=
\matris{cccc}{
\ze &\zz &\zz &\zz\\
\zz &\zz &\zz &\ze\\
\zz &\ze &\zz &\zz\\
\zz &\zz &\ze &\zz}
=
\matris{cccc}{
~~\zze &\zzz &\zzz &\zzz\\
\zzz &\zzz &\zzz &~~\zze\\
\zzz &~~\zze &\zzz &\zzz\\
\zzz &\zzz &~~\zze &\zzz
}
}
here written in both max-plus and traditional algebra notation for legibility purposes.} The similarity matrix $\CCb$ has the property of ``normalizing'' the $\MTD$ and $\MTDa$ matrices to a max-plus algebraic lower triangular form $\MTDn$ and a max-plus algebraic upper triangular form $\MTDan$ respectively:
\NEn{
\MTDn=\CCb \mtimes\MTD \mtimes\CCb^{T}\labeleq{similarityA}\\
\MTDan=\CCb \mtimes\MTDa \mtimes\CCb^{T}\labeleq{similarityB}
}
%
\examples{Taking the previous example of the trotting gait, the normalized matrices take the form
\NEn{\MTDn_{\mathrm{trot}}=\matris{cccc}{
\zz & \zz &\zz &\zz\\
\zz & \zz &\zz &\zz\\
\Td & \Td &\zz &\zz\\
\Td & \Td &\zz &\zz}\text{~~and~~}\MTDan_{\mathrm{trot}}=
\matris{cccc}{
\zz & \zz &\Td &\Td\\
\zz & \zz &\Td &\Td\\
\zz & \zz &\zz &\zz\\
\zz & \zz &\zz &\zz}
}
which are generated by the normal gait $\{1,2\}\mporder\{3,4\}$.}
Let $\#\leg{i}$ represent the number of elements of the set $\leg{i}$. For a general normal gait 
\NEn{\gaitn = \leg{1}\mporder\leg{2}\mporder\cdots\mporder\leg{\m}}
with $\Monexb{\card{i}, \card{j}}=\Monex{\#\leg{i}\times\#\leg{j}}$
the structure of the matrices $\MTDn$ and $\MTDan$ is:
\shrinkornot{
\NEn{\MTDn \ele \matris{cccccc}{\Mz & & & \cdots&~~\Mz\\
\Td\mtimes\Monexb{\card{2}, \card{1}} &\Mz&&&~~\vdots\\
\Mz &\Td\mtimes\Monexb{\card{3}, \card{2}}&\Mz&&\\
\vdots && \ddots&&\\
\Mz&\cdots&&  \Td\mtimes\Monexb{\card{m}, \card{m-1}}&~~\Mz
}\labeleq{pmatrix}
\\
\MTDan \ele \matris{cc}{
\Mz&\Td\mtimes\Monexb{\card{1}, \card{m}}\\
\Mz&\Mz
}\labeleq{qmatrix}
}
}
{
\NEn{\MTDn = \matris{cccccc}{\Mz & & & \cdots&~~\Mz\\
\Td\mtimes\Monexb{\card{2}, \card{1}} &\Mz&&&~~\vdots\\
\Mz &\Td\mtimes\Monexb{\card{3}, \card{2}}&\Mz&&\\
\vdots && \ddots&&\\
\Mz&\cdots&&  \Td\mtimes\Monexb{\card{m}, \card{m-1}}&~~\Mz
}\labeleq{pmatrix};
~~~~~~~~\MTDan = \matris{cc}{
\Mz&\Td\mtimes\Monexb{\card{1}, \card{m}}\\
\Mz&\Mz
}\labeleq{qmatrix}
}
}
From equation \refeq{pmatrix} it is clear that the matrix $\MTDn$ is always max-plus nilpotent, since the upper triangle is max-plus zero. For max-plus powers of $\MTDn$ we obtain:
\shrinkornot{
\NE{
&&\MTDn^{\otimes 2} = \matris{cccccc}{
\Mz &  & \cdots &    & \Mz \\
\Mz & ~\Mz &  &    &  \\
\Td\mtimes\Monexb{\card{3}, \card{1}}& ~\Mz & \Mz &  &   \vdots \\
\vdots &   \ddots &  & &  \\
\Mz &\!\!\! \cdots~~ &   \Td\mtimes\Monexb{\card{m}, \card{m-2}}~ & \Mz~ & \Mz},\\
&&\MTDn^{\otimes (m-1)} = \matris{cccc}{
\Mz & \cdots &   & \Mz       \\
\vdots &   &   &     \multirow{2}{*}{\vdots}\\
\Mz &  \Mz &   &       \\
\Td\mtimes\Monexb{\card{m}, \card{1}} & ~\Mz & \cdots & \Mz 
  },\\
&&\MTDn^{\otimes m}= \Mz
}}
{
\NE{
\MTDn^{\otimes 2} = \matris{cccccc}{
\Mz &  & \cdots &    & \Mz \\
\Mz & ~\Mz &  &    &  \\
\Td\mtimes\Monexb{\card{3}, \card{1}}& ~\Mz & \Mz &  &   \vdots \\
\vdots &   \ddots &  & &  \\
\Mz &\!\!\! \cdots~~ &   \Td\mtimes\Monexb{\card{m}, \card{m-2}}~ & \Mz~ & \Mz}\!\!;~
\MTDn^{\otimes (m-1)} = \matris{cccc}{
\Mz & \cdots &   & \Mz       \\
\vdots &   &   &     \multirow{2}{*}{\vdots}\\
\Mz &  \Mz &   &       \\
\Td\mtimes\Monexb{\card{m}, \card{1}} & ~\Mz & \cdots & \Mz 
  }\!\!;~
\MTDn^{\otimes m}= \Mz
}}
\begin{lemma}\labellem{nilpotency}
Max-plus nilpotency is invariant to max-plus similarity transformations (e.g. as defined in equations \refeq{similarityA}, \refeq{similarityB}).
\end{lemma}
\begin{proof}
Let $\bar{P}$ be nilpotent and $\bar{C}$ be a similarity matrix. Let $P=\bar{C}^{T}\otimes\bar{P}\otimes \bar{C}$. As such
\shrinkornot{
\NE{P^{\otimes p}\ele\left(\bar{C}^{T}\otimes \bar{P}\otimes \bar{C}\right)^{\otimes p}=\\
\ele \bar{C}^{T}\otimes \bar{P}\otimes \bar{C}\otimes \bar{C}^{T}\otimes \bar{P}\otimes \bar{C} \otimes \cdots \otimes \bar{C}^{T}\otimes \bar{P}\otimes \bar{C}\\
\ele \bar{C}^{T}\otimes \bar{P}^{\otimes p}\otimes \bar{C}}
}
{
\NEn{P^{\otimes p}=\left(\bar{C}^{T}\otimes \bar{P}\otimes \bar{C}\right)^{\otimes p}= \bar{C}^{T}\otimes \bar{P}\otimes \bar{C}\otimes \bar{C}^{T}\otimes \bar{P}\otimes \bar{C} \otimes \cdots \otimes \bar{C}^{T}\otimes \bar{P}\otimes \bar{C}
= \bar{C}^{T}\otimes \bar{P}^{\otimes p}\otimes \bar{C}}
}
and, 
\ZNE{\exists p_{0}> 0, \forall p\geq p_{0} :P^{\otimes p}= \bar{C}^{T}\otimes \bar{P}^{\otimes p}\otimes \bar{C}=\Mz.}\end{proof}
Given an arbitrary gait $\gait$ with associated matrices $\MTD$, $\MTDa$, $\Azero$, and $\Aone$ one can find the normal matrix $\MTDn$ which is max-plus nilpotent. From \reflem{nilpotency} then $\MTD$ is also max-plus nilpotent. In the beginning of \refsec{combinatorial} we have presented the synchronization equations \refeq{implicit} implicitly. However, if $\Azero^{*}$ exists then using equations (\ref{eq_max_plus_linear}) and \refeq{explicitsolution}, system \refeq{implicit} can be transformed into an explicit set of equations.
\begin{lemma}\labellem{a}
A sufficient condition for $\Azeros$ to exist is that the matrix $\MTD$ is nilpotent in the max-plus sense.
\end{lemma}
\begin{proof}
By direct computation, the repetitive products of $\Azero$ can be found to be
\NEn{\Azero^{\mtimes p}\!=\!\lista{lr}{\!\!\!\matris{c|c}{\Mz & \Tf^{\mtimes \frac{p+1}{2}} \mtimes \MTD^{\mtimes \frac{p-1}{2}}\\\hline
\Tf^{\mtimes \frac{p-1}{2}} \mplus \MTD^{\mtimes \frac{p+1}{2}} & \Mz} & \text{if $p$ is odd}\\[6mm]
\!\!\!\matris{c|c}{ \Tf^{\mtimes \frac{p}{2}} \mtimes \MTD^{\mtimes \frac{p}{2}}& \Mz\\\hline
\Mz&\Tf^{\mtimes \frac{p}{2}} \mtimes \MTD^{\mtimes \frac{p}{2}} } & \text{\!\!\!\!\!if $p$ is even}}}
If $\MTD$ is max-plus nilpotent, then there exists a finite positive integer $p_{0}$ such that $\forall p \geq p_{0} : \MTD^{\mtimes p} = \Mz \Rightarrow \Azero^{\mtimes (2 p+1)} = \Mz$, and therefore the max-plus sum for the computation of $\Azeros$ is finite:
\NEn{\Azeros=\bigoplus_{q=0}^{\infty}\Azero^{\mtimes q}=\bigoplus_{q=0}^{2p}\Azero^{\mtimes q}}
\end{proof}
Since $P$ is always max-plus nilpotent for gaits generated by expressions \refeq{Pdefinition} and \refeq{Qdefinition}, we conclude that $\Azeros$ is well defined. Let $\A$, that we call \emph{system matrix}, be defined by:
\NEn{\A=\Azero^{*}\otimes\Aone\labeleq{Adefinition}}
Equation \refeq{implicit} can be rewritten as:
\shrinkornot{
\NEn{\x(k)&=&\Azero\otimes\x(k)\oplus\Aone\otimes \x(k-1) \nonumber\\
&=&\Azero^{*}\otimes\Aone\otimes \x(k-1) \nonumber\\
&=&\A\mtimes\x(k-1).\labeleq{explicit}}}
{
\NEn{\x(k)=\Azero\otimes\x(k)\oplus\Aone\otimes \x(k-1) 
=\Azero^{*}\otimes\Aone\otimes \x(k-1) 
=\A\mtimes\x(k-1).\labeleq{explicit}}
}
For an arbitrary gait the internal structure of $\A$ can be quite complex. However, the gait $\gait$ associated to $\A$ can be transformed into a normal gait via a similarity transformation. Let
\NEn{
\CC=\matris{cc}{\CCb&\Mz\\\Mz&\CCb}
}
The similarity matrix $\CC$ transforms the system matrix $\A$ of an arbitrary gait $\gait$ into the system matrix $\An$ of a normal gait $\gaitn$ via the similarity transformation
\ZNE{
\An=\CC \mtimes\A \mtimes\CC^{T}.
}
This can be shown by direct computation:
\shrinkornot{
\NE{
&&\!\!\!\!\!\!\CC \mtimes\A \mtimes\CC^{T}=\CC \mtimes\Azero^{*}\mtimes\CC^{T}\mtimes\CC\otimes\Aone \mtimes\CC^{T}=\\[2mm]
\ele\matris{cc}{\CCb&\Mz\\\Mz&\CCb}\!\!\otimes\!\!
\matris{c|c}{\Mz &\Tf \mtimes \Mid\\\hline \MTD& \Mz  }^{*}\!\!\!\!\otimes\!\!\matris{cc}{\CCb&\Mz\\\Mz&\CCb}^{T}\!\!\otimes\\
&&\otimes\!\!\matris{cc}{\CCb&\Mz\\\Mz&\CCb}\!\!\otimes\!\!\matris{c|c}{\Mid &\Mz\\\hline \Tg \mtimes \Mid\mplus\MTDa& \Mid  }\!\!\otimes\!\!\matris{cc}{\CCb&\Mz\\\Mz&\CCb}^{T}=\\[3mm]
\ele
\bigoplus_{q=0}^{2m}\left(\matris{cc}{\CCb&\Mz\\\Mz&\CCb}\!\!\otimes\!\!
\matris{c|c}{\Mz &\Tf \mtimes \Mid\\\hline \MTD& \Mz  }\!\!\otimes\!\!\matris{cc}{\CCb&\Mz\\\Mz&\CCb}^{T}\right)^{\otimes q}\!\!\!\!\otimes\\
&&\otimes\!\!\matris{c|c}{\Mid & \Mz\\\hline \CCb \otimes \left(\Tg \mtimes \Mid\mplus\MTDa\right)\otimes \CCb^{T}& \Mid}=\\[3mm]
\ele
\left(\bigoplus_{q=0}^{2m}
{\underbrace{\matris{c|c}{\Mz &\Tf \mtimes \Mid\\\hline \MTDn& \Mz  }}_{\Azerob}}^{\otimes q}\right)\!\!\otimes
\!\!\underbrace{\matris{c|c}{\Mid & \Mz\\\hline \Tg \mtimes \Mid\mplus\MTDan& \Mid}}_{\Aoneb}=\\[2mm]
\ele \Azerosb\otimes\Aoneb=\Abar
}
}
{
\NEn{
&&\!\!\!\!\!\!\CC \mtimes\A \mtimes\CC^{T}=\CC \mtimes\Azero^{*}\mtimes\CC^{T}\mtimes\CC\otimes\Aone \mtimes\CC^{T}=\nonumber\\[2mm]
\ele\matris{cc}{\CCb&\Mz\\\Mz&\CCb}\!\!\otimes\!\!
\matris{c|c}{\Mz &\Tf \mtimes \Mid\\\hline \MTD& \Mz  }^{*}\!\!\!\!\otimes\!\!\matris{cc}{\CCb&\Mz\\\Mz&\CCb}^{T}\!\!\otimes\!\!\matris{cc}{\CCb&\Mz\\\Mz&\CCb}\!\!\otimes\!\!\matris{c|c}{\Mid &\Mz\\\hline \Tg \mtimes \Mid\mplus\MTDa& \Mid  }\!\!\otimes\!\!\matris{cc}{\CCb&\Mz\\\Mz&\CCb}^{T}=\nonumber\\[3mm]
\ele
\bigoplus_{q=0}^{2m}\left(\matris{cc}{\CCb&\Mz\\\Mz&\CCb}\!\!\otimes\!\!
\matris{c|c}{\Mz &\Tf \mtimes \Mid\\\hline \MTD& \Mz  }\!\!\otimes\!\!\matris{cc}{\CCb&\Mz\\\Mz&\CCb}^{T}\right)^{\otimes q}\!\!\!\!\otimes\!\!\matris{c|c}{\Mid & \Mz\\\hline \CCb \otimes \left(\Tg \mtimes \Mid\mplus\MTDa\right)\otimes \CCb^{T}& \Mid}=\nonumber\\[3mm]
\ele
\left(\bigoplus_{q=0}^{2m}
{\underbrace{\matris{c|c}{\Mz &\Tf \mtimes \Mid\\\hline \MTDn& \Mz  }}_{\Azerob}}^{\otimes q}\right)\!\!\otimes
\!\!\underbrace{\matris{c|c}{\Mid & \Mz\\\hline \Tg \mtimes \Mid\mplus\MTDan& \Mid}}_{\Aoneb}= \Azerosb\otimes\Aoneb=\Abar
}
}
Transforming an arbitrary gait into a normal gait is very useful since, by effectively switching rows and columns in $\A$, one obtains a very structured matrix $\An$ where analysis is simple. The interpretation of the similarity matrix $\CCb$ is that legs can be renumbered, simplifying algebraic manipulation. Besides max-plus nilpotency, other properties are invariant to similarity transformations: irreducibility is preserved since the graphs of $\A$ and $\An$ are equivalent up to a label renaming. Max-plus eigenvalues and eigenvectors are related by:
\shrinkornot{
\NE{
&&\A \otimes \ev = \eg \otimes \ev\\
&\Leftrightarrow& \CC\otimes \A\otimes \CC^{T}\otimes \CC\otimes \ev = \eg \otimes\CC\otimes \ev\\
&\Leftrightarrow& \Abar \otimes \bar{\ev} = \eg \otimes \bar{\ev},\text{~with~}\bar{\ev}= \CC\otimes \ev
}}
{
\NEn{
\A \otimes \ev = \eg \otimes \ev
\Leftrightarrow \CC\otimes \A\otimes \CC^{T}\otimes \CC\otimes \ev = \eg \otimes\CC\otimes \ev\Leftrightarrow \Abar \otimes \bar{\ev} = \eg \otimes \bar{\ev},\text{~with~}\bar{\ev}= \CC\otimes \ev
}}
\subsection{Structure of the system matrix $\bar{A}$}
Let:
\NEn       {\Tz=\Tf\mtimes\Td\text{~and~}\taugamma=\Tf\mtimes\Tg}
The structure of $\bar{A}$ can be obtained via a laborious but straightforward set of algebraic manipulations. For an arbitrary gait $\gait$ we compute the normal gait $\gaitn$ via the similarity transformation with the matrix $\CC$. By observing the structures of $\An_{0}$ and $\An_{1}$ (derived from $\MTDn$ and $\MTDan$) a closed-form solution can be obtained for $\An_{0}^{*}$:
\NEn{\An_{0}^{*}=\matris{cc}{\Md&~~~~\Tf\mtimes\Md\\\Mdd&~~~~\Md}
}
where $\Md=(\Tf\mtimes \MTDn)^{*}$, illustrated in equation \refeq{deltamatrix} on page \pageref{eq.deltamatrix}. The matrix $\Mdd$ is defined in equation \refeq{Mdddefinition} again on page \pageref{eq.Mdddefinition}. Note that \mbox{$\Tf\mtimes\Mdd\mplus\Mid=\Md$} and $\Md \geq \Mdd$. An expression for $\An$ is then obtained:
\shrinkornot{
\NEn{
\An\ele\An_{0}^{*}\mtimes\An_{1}\nonumber\\
\ele\matris{cc}{\Md&~~~~\Tf\mtimes\Md\\\Mdd&~~~~\Md}\mtimes\matris{cc}{\Mid&~~~~\Mz\\\Tg\mtimes\Mid\mplus\MTDan&~~~~\Mid}\labeleq{xxxxxxxxx}
\\
\ele\matris{cc}
{\Md\mplus\Tf\mtimes\Tg\mtimes\Md\mplus\Tf\mtimes\Md\mtimes\MTDan  &~~~~ \Tf\mtimes\Md\\
\Mdd\mplus\Tg\mtimes\Md\mplus\Md\mtimes\MTDan  & ~~~~\Md
}
\nonumber}
}
{
\NEn{
\An\ele\An_{0}^{*}\mtimes\An_{1}=\matris{cc}{\Md&~~~~\Tf\mtimes\Md\\\Mdd&~~~~\Md}\mtimes\matris{cc}{\Mid&~~~~\Mz\\\Tg\mtimes\Mid\mplus\MTDan&~~~~\Mid}\labeleq{xxxxxxxxx}
\\
\ele\matris{cc}
{\Md\mplus\Tf\mtimes\Tg\mtimes\Md\mplus\Tf\mtimes\Md\mtimes\MTDan  &~~~~ \Tf\mtimes\Md\\
\Mdd\mplus\Tg\mtimes\Md\mplus\Md\mtimes\MTDan  & ~~~~\Md
}
\nonumber}
}
Let $\Mv=\Md\mtimes\MTDan$, as illustrated by equation \refeq{vmatrix}. One can show that: 
\NEn{
&&\Md\mtimes\Md=\Md\labeleq{minusa}\\ 
&&\Md\mtimes\Mv=\Mv\\
&&\Mv\mtimes\Mv=\Tdf^{\mtimes(m-1)}\mtimes\Td\mtimes\Mv\labeleq{minus}
}
Since $\mu\mtimes\Md\geq\Md$ for any $\mu>0$, and $\Md \geq \Mdd$, expression \refeq{xxxxxxxxx} simplifies to:
\NEn{
\An\ele\matris{cc}
{\Tf\mtimes(\Tg\mtimes\Md\mplus\Mv)  & ~~~~\Tf\mtimes\Md\\
\Tg\mtimes\Md\mplus\Mv  & ~~~~\Md
}\labeleq{irreducible}
.}
\shrinkornot{
\NE{\text{\!Let~}t_{\leg{i}}(k)\ele\matris{cccc}{t_{[\leg{i}]_{1}}(k)&t_{[\leg{i}]_{2}}(k)&\cdots&t_{[\leg{i}]_{\#\leg{i}}}(k)}^{T}\\
\text{and~}l_{\leg{i}}(k)\ele\matris{cccc}{l_{[\leg{i}]_{1}}(k)&l_{[\leg{i}]_{2}}(k)&\cdots&l_{[\leg{i}]_{\#\leg{i}}}(k)}^{T}\!\!.}}
{
\NE{\text{\!Let~}t_{\leg{i}}(k)=\matris{cccc}{t_{[\leg{i}]_{1}}(k)&t_{[\leg{i}]_{2}}(k)&\cdots&t_{[\leg{i}]_{\#\leg{i}}}(k)}^{T}
\text{and~}l_{\leg{i}}(k)=\matris{cccc}{l_{[\leg{i}]_{1}}(k)&l_{[\leg{i}]_{2}}(k)&\cdots&l_{[\leg{i}]_{\#\leg{i}}}(k)}^{T}\!\!.}}
Equations \refeq{verybigAexplicitStart}--\refeq{verybigAexplicit} illustrate the resulting structure of $\bar{A}$ written in the system form $\bar{x}(k)=\bar{A}\otimes\bar{x}(k-1)$, with $\bar{x}(k)=C\otimes x(k)$, and $\Midxb{i}=\Midx{\#\leg{i}}$. In equation \refeq{verybigAexplicit}, the $\otimes$ operator is omitted in unambiguous locations due to space limitations.

\renewcommand{\card}[1]{#1}
\newcommand{\cardtimes}{,}
\newcommand{\diaga}{\taugamma\mmtimes}
\newcommand{\diagb}{\Tf\mmtimes}
\newcommand{\diagc}{\Tg\mmtimes}
\newcommand{\diagd}{}
\newcommand{\mmtimes}{}
\newcommand{\pmmtimes}{\mtimes}
\newcommand{\mmplus}{+}
\newcommand{\hide}[1]{}

\shrinkornot{\newcommand{\schrink}{\!\!}}
{\newcommand{\schrink}{\!\!\!}}

\newcommand{\rotate}[2]{\mathrel{\rotatebox[origin=c]{#1}{$#2$}}}

\shrinkornot{\newcommand{\elle}{&\!\!\!\!=\!\!\!\!&}}
{\newcommand{\elle}{&\!\!\!\!\!\!=\!\!\!\!\!\!&}}

\begin{figure*}
\shrinkornot{\normalsize}{\footnotesize}
\NEn{\Mdd =  \Td\mtimes\matris{cccccc}{\Mz & &~~~~~~~~~ & \cdots~~~~~~&\Mz\\
\Monexb{\card{2} \cardtimes \card{1}} &\Mz&&&\vdots\\
\Tdf\mtimes\Monexb{\card{3} \cardtimes \card{1}} &\Monexb{\card{3} \cardtimes \card{2}}&\Mz&&\\
\vdots &&\rotate{15}{\ddots}&\rotate{15}{\ddots}~~~~~~&\\
\Tdf^{\mtimes (m-2)}\mtimes\Monexb{\card{m} \cardtimes \card{1}} &\cdots& \Tdf\mtimes\Monexb{\card{m} \cardtimes \card{m-1}}&~~~ \Monexb{\card{m} \cardtimes \card{m-2}}~~~&\Mz}\labeleq{Mdddefinition}
}
\shrinkornot{}{\\[-10mm]}
\NEn{\Md = \matris{cccccc}{\Midxb{\card{1}} & & ~~~~~~~~~& \cdots~~~~&\Mz\\
\Tdf\mtimes\Monexb{\card{2}\cardtimes \card{1}} &\Midxb{\card{2}}&&&\vdots\\
\Tdf^{\mtimes 2}\mtimes\Monexb{\card{3}\cardtimes \card{1}} &\Tdf\mtimes\Monexb{\card{3}\cardtimes \card{2}}&\Midxb{\card{3}}&&\\
\vdots &&\rotate{22}{\ddots}&\rotate{22}{\ddots}~~~~&\\
\Tdf^{\mtimes (m-1)}\mtimes\Monexb{\card{m}\cardtimes \card{1}} &\cdots& \Tdf^{\mtimes 2}\mtimes\Monexb{\card{m}\cardtimes \card{m-2}}&~~~ \Tdf\mtimes\Monexb{\card{m}\cardtimes \card{m-1}}~~~&\Midxb{\card{m}}
}\labeleq{deltamatrix}
}
\shrinkornot{}{\\[-10mm]}
\NEn{\Mv=\matris{c|c}{\Mz_{n\cardtimes (n-\card{m})}&\begin{array}{c} 
\Td\mtimes\Monexb{\card{1}\cardtimes \card{m}} \\
\Td\mtimes\Tdf\mtimes\Monexb{\card{2}\cardtimes \card{m}} \\
\vdots\\
\Td\mtimes\Tdf^{\mtimes(m-1)}\mtimes\Monexb{\card{m}\cardtimes \card{m}} \\
\end{array}
}\labeleq{vmatrix}
}
\shrinkornot{}{\\[-15mm]}
\NEn{\bar{x}(k)\elle\bar{A}\otimes\bar{x}(k-1)\Leftrightarrow\labeleq{verybigAexplicitStart}
\\
\bar{x}(k)\elle\matris{c||c}
{\Tf\mtimes(\Tg\mtimes\Md\mplus\Mv)  ~~& ~~\Tf\mtimes\Md\\\hline\hline
\Tg\mtimes\Md\mplus\Mv  ~~& ~~\Md
}\otimes\bar{x}(k-1) \Leftrightarrow\labeleq{verybigAexplicitA}
\\
\matris{c}{
\schrink\tT{\leg{1}}(k)\schrink\\\vdots\\\hline \schrink\tT{\leg{m}}(k)\schrink\\\hline\hline
\schrink\tL{\leg{1}}(k)\schrink\\\vdots\\\hline \schrink\tL{\leg{m}}(k)\schrink\
}
\elle
\matris{c|c||c}
{
\multirow{2}{*}{$A_{11}$} & \multirow{2}{*}{$A_{12}$} & \multirow{2}{*}{$A_{13}$}\\&&\\\hline
 ~~~~~A_{21}~~~~~& A_{22} & ~~~~~A_{23}~~~~~\\\hline\hline
\multirow{2}{*}{$A_{31}$} & \multirow{2}{*}{$A_{32}$} & \multirow{2}{*}{$A_{33}$}\\&&\\\hline
A_{41} & A_{42} & A_{43}
}
\schrink\otimes\schrink
\matris{c}{
\schrink\tT{\leg{1}}(k-1)\schrink\\\vdots\\\hline \schrink\tT{\leg{m}}(k-1)\schrink\\\hline\hline
\schrink\tL{\leg{1}}(k-1)\schrink\\\vdots\\\hline \schrink\tL{\leg{m}}(k-1)\schrink\
}\Leftrightarrow\labeleq{verybigAexplicitB}
\\
\underbrace{\matris{c}{\schrink\tT{\leg{1}}(k)\schrink\\[0.8mm] \schrink\tT{\leg{2}}(k)\schrink\\[0.8mm] \schrink\vdots\schrink  \\[0.8mm] \hline\schrink\tT{\leg{m}}(k)\schrink\\[0.8mm]\hline\hline \schrink\tL{\leg{1}}(k)\schrink\\[0.8mm] \schrink\tL{\leg{2}}(k)\schrink\\[0.8mm]\vdots  \\[0.8mm] \hline\schrink \tL{\leg{m}}(k)\schrink}}_{\bar{x}(k)}
\elle
\underbrace{\matris{ccc|c||ccc}{
\diaga\Midxb{\card{1}} &   \cdots&\Mz&\Tdf\mmtimes\Monexb{\card{1}\cardtimes \card{m}}&
\diagb\Midxb{\card{1}} &  \cdots&\Mz\\
\diaga\Tdf\mmtimes\Monexb{\card{2}\cardtimes \card{1}} &\schrink\diaga\Midxb{\card{2}}\schrink&\vdots&\Tdf^{\pmmtimes 2}\mmtimes\Monexb{\card{2}\cardtimes \card{m}}&
\diagb\Tdf\mmtimes\Monexb{\card{2}\cardtimes \card{1}} &\schrink\diagb\Midxb{\card{2}}\schrink&\vdots\\
\vdots &&\ddots&\vdots&
\vdots &\ddots&\\\hline
\schrink\diaga\Tdf^{\pmmtimes (m-1)}\mmtimes\Monexb{\card{m}\cardtimes \card{1}} \schrink&\schrink\cdots\schrink&\schrink\diaga \Tdf\mmtimes\Monexb{\card{m}\cardtimes \card{m-1}}&\schrink\diaga\Midxb{\card{m}}\mplus\Tdf^{\pmmtimes m}\mmtimes\Monexb{\card{m}\cardtimes \card{m}}\schrink&
\diagb\Tdf^{\pmmtimes (m-1)}\mmtimes\Monexb{\card{m}\cardtimes \card{1}}\schrink &\cdots&\schrink\diagb\Midxb{\card{m}}\schrink\\\hline\hline
\diagc\Midxb{\card{1}} &   \cdots&\Mz&\Td\mmtimes\Monexb{\card{1}\cardtimes \card{m}}&
\diagd\Midxb{\card{1}} &  \cdots&\Mz\\
\diagc\Tdf\mmtimes\Monexb{\card{2}\cardtimes \card{1}} &\schrink\diagc\Midxb{\card{2}}\schrink&\vdots&\Td\mmtimes\Tdf\mmtimes\Monexb{\card{2}\cardtimes \card{m}}&
\Tdf\mmtimes\Monexb{\card{2}\cardtimes \card{1}} &\diagd\Midxb{\card{2}}&\vdots\\
\vdots &&\ddots&\vdots&
\vdots &\ddots&\\\hline
\schrink\diagc\Tdf^{\pmmtimes (m-1)}\mmtimes\Monexb{\card{m}\cardtimes \card{1}} \schrink&\schrink\cdots\schrink&\schrink\diagc \Tdf\mmtimes\Monexb{\card{m}\cardtimes \card{m-1}}\schrink&\diagc\Midxb{\card{m}}\mplus\Td\mmtimes\Tdf^{\pmmtimes(m-1)}\mmtimes\Monexb{\card{m}\cardtimes \card{m}}&
\Tdf^{\pmmtimes (m-1)}\mmtimes\Monexb{\card{m}\cardtimes \card{1}} &\cdots&\Midxb{\card{m}}
}}_{\bar{A}}
\schrink\mtimes\schrink
\underbrace{\matris{c}{\schrink\tT{\leg{1}}(k-1)\schrink\\[0.8mm] \schrink\tT{\leg{2}}(k-1)\schrink\\[0.8mm] \vdots \\[0.8mm] \hline\schrink\tT{\leg{m}}(k-1)\schrink\\[0.8mm]\hline\hline \schrink\tL{\leg{1}}(k-1)\schrink\\[0.8mm] \schrink\tL{\leg{2}}(k-1)\schrink\\[0.8mm]\vdots \\[0.8mm]\hline \schrink\tL{\leg{m}}(k-1)\schrink}}_{\bar{x}(k-1)}
\labeleq{verybigAexplicit}}
%
%
\end{figure*}

%
%

\subsection{Max-plus eigenstructure of the system matrix}\labelsec{eigen}
\labelsec{eigenstructure}

Given a parameterization of a gait, it is fundamental to understand whether the system $x(k)=A\otimes x(k-1)$ reaches a unique steady state behavior. In robotics this is the equivalent of asking ``does the robot walk/run as specified? Is it robust to disturbances?''. These questions are answered by analyzing the max-plus eigenstructure of the system matrix: a unique eigenvalue means that the legs have a unique cycle time, and a unique (up to scaling) eigenvector means that the legs always reach the same motion pattern, independently of the initial condition or disturbances. The results obtained below use various analysis techniques available for max-plus linear systems. This is necessary due to the intrinsic time structure associated with the problem. Petri net tools (e.g. incidence matrices) can be used to understand structural properties of the system, such as irreducibility, but temporal properties are better analyzed using max-plus linear tools. In \refsec{graph} we show that for a fixed structure (i.e. a single Petri net) unique or non-unique eigenvectors are found by changing the holding time parameters. This result could not be captured by the Petri net structure alone. The analysis steps presented from here on are summarized in \reffig{layout}. 

Consider the following assumption (which is always satisfied in practice since the leg swing and stance times are always positive numbers):\\[-2mm]
\begin{assumption}[A1] 
$\Tg,\Tf>0$\\[-1mm]
\end{assumption}
\begin{lemma}\labelth{existanceuniqueness}
If assumption A1 is satisfied then 
\NEn{\eg \bydefinition  \Mpower{\Tz}{\m}\mplus \taugamma\labeleq{eigenvaluedefinition}}
 is a max-plus eigenvalue of the system matrix $\A$ (and $\Abar$) defined by equations \refeq{explicit} (and \refeq{verybigAexplicitStart}--\refeq{verybigAexplicit}), and $\ev\in\Rmax^{2n}$ defined by
\NEn{
\forall j\in\{1,\dots,\m\}, \forall q\in \leg{j}:&&[\ev]_{q}\bydefinition\Tf\mtimes \Mpower{\Tz}{j-1}\labeleq{eigenvectordefinitionA}\\
&&[\ev]_{q+n}\bydefinition\Mpower{\Tz}{j-1}\labeleq{eigenvectordefinitionB}
}
is a max-plus eigenvector of $A$.
\end{lemma}
\begin{proof}
With $\evb\in\Rmax^{n}$, let 
$[\evb]_{q}=[\ev]_{q+n}$ for all $j$ and $q\in \leg{j}$. Then  
$\ev=[(\Tf\mtimes\evb)^T~~\evb^T]^T$. 
Recall equations (\ref{eq_max_plus_linear}) and \refeq{explicitsolution} with new variables $z$ and $B$ such that $z= B\otimes z\oplus b$ with solution $z= B^{*}\otimes b$. Now let $z=\eg\otimes\ev$, $B=\Azero$, and $b=\Aone\otimes \ev$. We obtain
\shrinkornot{
\NE{\eg\otimes\ev \ele \Azero \otimes \eg\otimes\ev \oplus \Aone\otimes \ev\\
 \ele \Azeros \otimes \Aone\otimes \ev=\A\otimes \ev
}}
{
\NEn{\eg\otimes\ev \ele \Azero \otimes \eg\otimes\ev \oplus \Aone\otimes \ev \Azeros \otimes \Aone\otimes \ev=\A\otimes \ev
}}
Given the previous result, it is sufficient to show that if $\eg$ and $\ev$ are a max-plus eigenvalue and eigenvector of $\A$ respectively, then replacing the state variable $\x(k-1)$ by $\ev$ and $\x(k)$ by $\eg\mtimes\ev$ in equation \refeq{sssync} holds true:
\shrinkornot{
\NE{
&&\eg\mtimes\ev = \eg\mtimes
\matris{c}{\Tf\mtimes\evb\\\evb} =\\
&&\eg\mtimes\matris{c|c}{\Mz &\Tf \mtimes \Mid\\\hline \MTD& \Mz  }\mtimes
\ev \mplus
\matris{c|c}{\Mid &\Mz\\\hline \Tg \mtimes \Mid\mplus\MTDa& \Mid  }\mtimes
\ev=\\
&&
\matris{c|c}{\Mid & \eg\mtimes\Tf\mtimes\Mid\\\hline 
\eg \mtimes \MTD\mplus \Tg \mtimes \Mid\mplus\MTDa &\Mid}\mtimes
\matris{c}{\Tf\mtimes\evb\\\evb}
}}
{
\NEn{
\nonumber&&\eg\mtimes\ev = \eg\mtimes
\matris{c}{\Tf\mtimes\evb\\\evb} =\eg\mtimes\matris{c|c}{\Mz &\Tf \mtimes \Mid\\\hline \MTD& \Mz  }\mtimes
\ev \mplus
\matris{c|c}{\Mid &\Mz\\\hline \Tg \mtimes \Mid\mplus\MTDa& \Mid  }\mtimes
\ev=\\
&&
\matris{c|c}{\Mid & \eg\mtimes\Tf\mtimes\Mid\\\hline 
\eg \mtimes \MTD\mplus \Tg \mtimes \Mid\mplus\MTDa &\Mid}\mtimes
\matris{c}{\Tf\mtimes\evb\\\evb}
}}
The previous expression is equivalent to the following two equations:
\NEn{\eg \mtimes\Tf \mtimes \evb \ele \Tf\mtimes \evb \mplus \eg \mtimes\Tf \mtimes \evb \labeleq{twoequationsA}\\
\eg \mtimes \evb \ele\Tf\mtimes(\eg \mtimes \MTD\mplus \Tg \mtimes \Mid\mplus\MTDa )\mtimes\evb\mplus\evb\labeleq{twoequationsB}
}
Since $\eg>0$ (by assumption A1), equation \refeq{twoequationsA} is always verified. Thus we focus on equation \refeq{twoequationsB}, which can be simplified due to $\Tf\mtimes\Tg>0$:
\NEn{\eg \mtimes \evb \ele(\Tf\mtimes\Tg) \mtimes \evb\mplus\Tf\mtimes(\eg \mtimes \MTD\mplus \MTDa )\mtimes\evb\labeleq{zc}}
Let $\Td\mtimes\MTDz=\MTD$ and $\Td\mtimes\MTDaz=\MTDa$, i.e., all entries of matrices $\MTDz$ and $\MTDaz$ are either $\ze$ or $\zz$ to obtain (recall that $\Tz=\Tf\mtimes\Td$ and $\taugamma=\Tf\mtimes\Tg$):
\NEn{\eg \mtimes \evb \ele \taugamma\mtimes \evb\mplus\Tz\mtimes (\eg\mtimes \MTDz\mplus  \MTDaz )\mtimes\evb \labeleq{zb}}
We now consider two cases:\\\\
i) First we analyze the row indices of equation \refeq{zb} that are elements of the sets $\leg{2},\dots,\leg{m}$. For each $j\in\{1,\dots,m-1\}$ and for each row $p \in \leg{j+1}$ we obtain (notice that according to \refeq{Qdefinition} all the elements of $\left[\MTDaz\right]_{p,\cdot}$ are $\zz$ since $p \notin \leg{1}$, and that $\left[\evb\right]_{p}=\Mpower{\Tz}{j}$ for $p\in\leg{j+1}$):
\NEn{
\left[\eg\mtimes\evb\right]_{p} \ele \left[ \taugamma\mtimes \evb\right]_{p}\mplus \Tz\mtimes\left[\eg \mtimes\MTDz\mplus \MTDaz \mtimes \evb\right]_{p} \Leftrightarrow\\
\eg\mtimes\left[\evb\right]_{p} \ele\taugamma\mtimes\left[\evb\right]_{p}\mplus \Tz\mtimes\left[\eg\mtimes\MTDz\right]_{p,\cdot}\mtimes \evb\mplus  \underbrace{\left[\MTDaz\right]_{p,\cdot}}_{\zz} \mtimes \evb\Leftrightarrow\\
\eg\mtimes\Mpower{\Tz}{j} \ele\taugamma\mtimes\Mpower{\Tz}{j}\mplus \Tz\mtimes \bigoplus_{q\in\leg{j}}
\eg\mtimes\underbrace{\left[\MTDz\right]_{p,q}}_{\ze}\mtimes\left[\evb\right]_{q} \Leftrightarrow\\
\eg\mtimes\Mpower{\Tz}{j} \ele \taugamma\mtimes\Mpower{\Tz}{j}\mplus\Tz\mtimes \eg\mtimes \Mpower{\Tz}{j-1} \Leftrightarrow\\
\eg\mtimes\Mpower{\Tz}{j} \ele \taugamma\mtimes\Mpower{\Tz}{j}\mplus\eg\mtimes\Mpower{\Tz}{j}
}
The last term always holds true since $\eg \geq \taugamma$.
Thus for rows $p\in \leg{2},\dots,\leg{m}$ equation \refeq{zb} holds true.\\
ii) We now look at all the remaining rows $p$ such that $p \in \leg{1}$ (noticing now that according to \refeq{Pdefinition} all the elements of $\left[\MTDz\right]_{p,\cdot}$ are $\zz$ and that $\left[\evb\right]_{p}= \ze$ since $p \in \leg{1}$):
\NEn{
\left[\eg\mtimes\evb\right]_{p} \ele \left[\taugamma\mtimes\evb\right]_{p}\mplus\Tz\mtimes\left[\eg \mtimes\MTDz\mplus \MTDaz \right]_{p,\cdot}\mtimes \evb \Leftrightarrow\\
\eg\mtimes\left[\evb\right]_{p} \ele \taugamma\mplus\Tz\mtimes\underbrace{\left[\eg \mtimes\MTDz\right]_{p,\cdot}}_{\zz}\mtimes \hspace{0.8mm} \evb\mplus  \Tz\mtimes\left[\MTDaz\right]_{p,\cdot} \mtimes \evb\Leftrightarrow\\
\eg \ele \taugamma\mplus \Tz\mtimes \bigoplus_{q\in\leg{m}}
\underbrace{\left[\MTDaz\right]_{p,q}}_{\ze}\mtimes\left[\evb\right]_{q} \Leftrightarrow\\
\eg\ele \taugamma\mplus \Tz\mtimes \Mpower{\Tz}{\m-1} \Leftrightarrow\\
\eg\ele \taugamma\mplus \Mpower{\Tz}{\m}
}
Combining i) and ii) we conclude that equation \refeq{zb} holds true.
\end{proof}

\examples{
Consider again the trotting gait for a quadruped $\gait_{\mathrm{trot}}$ defined in \refeq{trottinggait}. For this gait $m=2$, resulting in:
\shrinkornot{
\NE{
\ev_{\mathrm{trot}}=\matris{c}{\Tf \\ \Td\otimes\Tf^{\otimes 2}
\\\Td\otimes\Tf^{\otimes 2}\\\Tf \\ 0 \\ \Td\otimes \Tf\\\Td\otimes \Tf\\0}
;~~~~~\eg_{\mathrm{trot}}= (\Tf\otimes\Td)^{\otimes 2}\oplus \Tf\otimes\Tg
}
}
{\NE{
\ev_{\mathrm{trot}}=\matris{cccccccc}{\Tf & (\Td\otimes\Tf^{\otimes 2})
&(\Td\otimes\Tf^{\otimes 2})&\Tf & 0 & (\Td\otimes \Tf)&(\Td\otimes \Tf)&0}^T
\!;~~\eg_{\mathrm{trot}}= (\Tf\otimes\Td)^{\otimes 2}\oplus \Tf\otimes\Tg
}
}
}
\begin{lemma}\labellem{irreducibility}
Matrices $A$ and $\bar{A}$ are irreducible.
\end{lemma}
\begin{proof}
The sub-matrices $A_{12},A_{22},A_{32},A_{42},A_{21}$ defined in expressions \refeq{verybigAexplicitB} and \refeq{verybigAexplicit} have all their elements different from $\zz$. The sub-matrix $A_{23}$ has all diagonal elements different from $\zz$. As such, any node can be reached by any other node via the rows defined by $A_{12},A_{22},A_{32},A_{42}$ and the columns defined by $A_{21},A_{22},A_{23}$. Therefore $\bar{A}$ is irreducible. Since $A$ is a similarity transformation away from $\bar{A}$ then we conclude that $A$ is also irreducible.
\end{proof}
\begin{corollary}
The max-plus eigenvalue $\eg$ of $A$ (and $\Abar$) given by \refeq{eigenvaluedefinition} is unique.
\end{corollary}
The max-plus eigenvector $v$ defined by \refeq{eigenvectordefinitionA}--\refeq{eigenvectordefinitionB} is not necessarily unique, given assumption A1 alone. Since, to the authors' best knowledge, there exists no algebraic method to prove max-plus eigenvector uniqueness in general, we take advantage of the precedence graph of $\bar{A}$ to further investigate this property. If the critical graph of a irreducible max-plus system matrix has a single strongly connected subgraph, then its max-plus eigenvector is unique up to a max-plus scaling factor (see \cite{baccelli92}, Theorem 3.101). We proceed by computing the critical graph(s) of $\bar{A}$.

\subsection{The precedence graph of $\bar{A}$}\labelsec{graph}

Given expression \refeq{verybigAexplicit} it is possible to construct the precedence graph of $\bar{A}$. Since this graph can be quite large for a general $\bar{A}$, we find it more efficient to first group ``similar'' nodes into a single node, i.e. apply a procedure called node reduction (\reffig{GraphReductions}). Next, we show various subgraphs of the graph of $\bar{A}$ to better illustrate its structure (\reffig{GraphConstruction}). The total precedence graph of $\bar{A}$ is thus the combination of Figures \ref{fg.GraphReductions} and \ref{fg.GraphConstruction}.
\Gshrink
The process of constructing the graph of $\bar{A}$ starts by grouping all nodes of an event associated with a group of legs $\leg{i}$ into a single node. This can be accomplished since event nodes from the same group of legs $\leg{i}$ have ``similar'' incoming and outgoing arcs. As an example, consider the first set of $\#\leg{1}$ rows of $\bar{A}$ as defined in expression \refeq{verybigAexplicit}:
\shrinkornot{
\NEn{
t_{\leg{1}}(k)\ele\taugamma\otimes E_{1} \otimes t_{\leg{1}}(k-1)\oplus\nonumber\\
&&\Tdf\otimes \Mone_{1,m}\otimes t_{\leg{m}}(k-1)\oplus\labeleq{nodereductionexample}\\
&&\Tf\otimes E_{1}\otimes l_{\leg{1}}(k-1)\nonumber
}}
{
\NEn{
t_{\leg{1}}(k)\ele\taugamma\otimes E_{1} \otimes t_{\leg{1}}(k-1)\oplus\Tdf\otimes \Mone_{1,m}\otimes t_{\leg{m}}(k-1)\oplus
\Tf\otimes E_{1}\otimes l_{\leg{1}}(k-1)\labeleq{nodereductionexample} 
}}
The precedence graph for equation \refeq{nodereductionexample} consists of $3 \times \#\leg{1}$ nodes, since it involves the vectors $t_{\leg{1}}$, $t_{\leg{m}}$, and $l_{\leg{1}}$. The relation between $t_{\leg{1}}(k)$ and $t_{\leg{1}}(k-1)$ results in $\#\leg{1}$ self connected arcs in the $t_{\leg{1}}$ events with weights $\taugamma$. Instead of expressing all elements of $t_{\leg{1}}$ as individual nodes with self arcs, we reduce then to a single node with one self arc, as seen in \reffig{GraphReductions}-a2. The dashed attribute used on the self arc indicates that for each node in the group only self arcs exist, as expressed by the ``connecting'' matrix $E_{1}$. The relation between $t_{\leg{1}}(k)$ and $t_{\leg{m}}(k-1)$ is somewhat more involved, since it contains $\#\leg{1}\times \#\leg{m}$ arcs, as expressed by the connecting matrix $\Mone_{1,m}$. The resulting node reduction is illustrated in \reffig{GraphReductions}-b1. The node reduction for the relation between $t_{\leg{1}}$ and $l_{\leg{1}}$ is illustrated in  \reffig{GraphReductions}-a4. Again we use dashed attributes on the arcs to represent the connecting matrix $E_{1}$. For all other relations with connecting matrices $\Mone$ we use solid arcs. We make an exception in Figures \ref{fg.GraphReductions}-c1 to \ref{fg.GraphReductions}-c4 where different line attributes are used to distinguish arcs from $t_{\leg{p}}\to t_{\leg{q}}$, $t_{\leg{p}}\to l_{\leg{q}}$, etc. The same line attributes are used in Figures \ref{fg.GraphConstruction}-c1 and \ref{fg.GraphConstruction}-c2. Note that multiple incoming arcs to a node are related via the $\mplus$ operation, e.g. as in the example \refeq{nodereductionexample} the node $t_{\leg{1}}$ has 3 incoming arcs, illustrated in Figure \ref{fg.GraphConstruction}.
The following list summarizes the node reduction:
\begin{itemize}
\item \reffig{GraphReductions}-a1 illustrates node reduction of the term $\Tdf^{\otimes m}\otimes \Mone_{m,m}$ of sub-matrix $A_{22}$ from expressions  \refeq{verybigAexplicitB} and \refeq{verybigAexplicit}
\item \reffig{GraphReductions}-a2 illustrates the node reduction of the block diagonal of matrix $A_{11}$ and the $\taugamma\otimes \bar{E}_m$ term of $A_{22}$.
\item \reffig{GraphReductions}-a3 illustrates the node reduction of the block diagonal of matrix $[A_{33}^{T}~~A_{43}^{T}]^{T}$
\item \reffig{GraphReductions}-a4 illustrates the node reduction of the term $\Tg\otimes E_{m}$ of sub-matrix $A_{42}$ together with the block diagonals of matrices $A_{31}$ and $[A_{13}^{T}~~A_{23}^{T}]^{T}$
\item Figures \ref{fg.GraphReductions}-b1 and \ref{fg.GraphReductions}-b2 illustrate the node reduction for the columns formed by the matrices (not including the term $\Tg\otimes E_{m}$ from matrix $A_{42}$ already represented in Figure \ref{fg.GraphReductions}-a4) $A_{12}$ and $[A_{32}^{T}~~A_{42}^{T}]^{T}$ respectively
\item Figures \ref{fg.GraphReductions}-c1 to \ref{fg.GraphReductions}-c4 illustrate the node reduction of the off-diagonal elements of matrices $\taugamma\otimes W$, $\Tf\otimes W$, $\Tg\otimes W$, and $W$, from expression \refeq{verybigAexplicitA} respectively.
\end{itemize}
Given the node reduction one can now proceed to construct the precedence graph of $\bar{A}$:
\begin{itemize}
\item \reffig{GraphConstruction}-a is the graph of the block diagonal of $\bar{A}$ together with the block diagonals of the sub-matrices $\matris{cc}{A_{31}&A_{32}\\A_{41}&A_{42}}$ and $[A_{13}^{T}~~A_{23}^{T}]^T$ using the node reductions presented in Figures \ref{fg.GraphReductions}-a1 to \ref{fg.GraphReductions}-a4.
\item \reffig{GraphConstruction}-b is the graph of the columns formed by the matrices $A_{12}$ and $[A_{32}^{T}~~A_{42}^{T}]^{T}$ using node reductions presented in Figures \ref{fg.GraphReductions}-b1 and \ref{fg.GraphReductions}-b2.
\item Figures \ref{fg.GraphConstruction}-c1 and \ref{fg.GraphConstruction}-c2 illustrate two subgraphs of the remaining columns of $\bar{A}$. Note that we only present the subgraphs of the first sets of $\#\leg{1}$ and $\#\leg{2}$ out of a total of $m-1$ columns. These follow the same pattern. We use different attributes on the arcs, such as dashed, thick solid, etc., to distinguish the different node reductions, as presented in Figures \ref{fg.GraphReductions}-c1 to \ref{fg.GraphReductions}-c4.
\end{itemize}

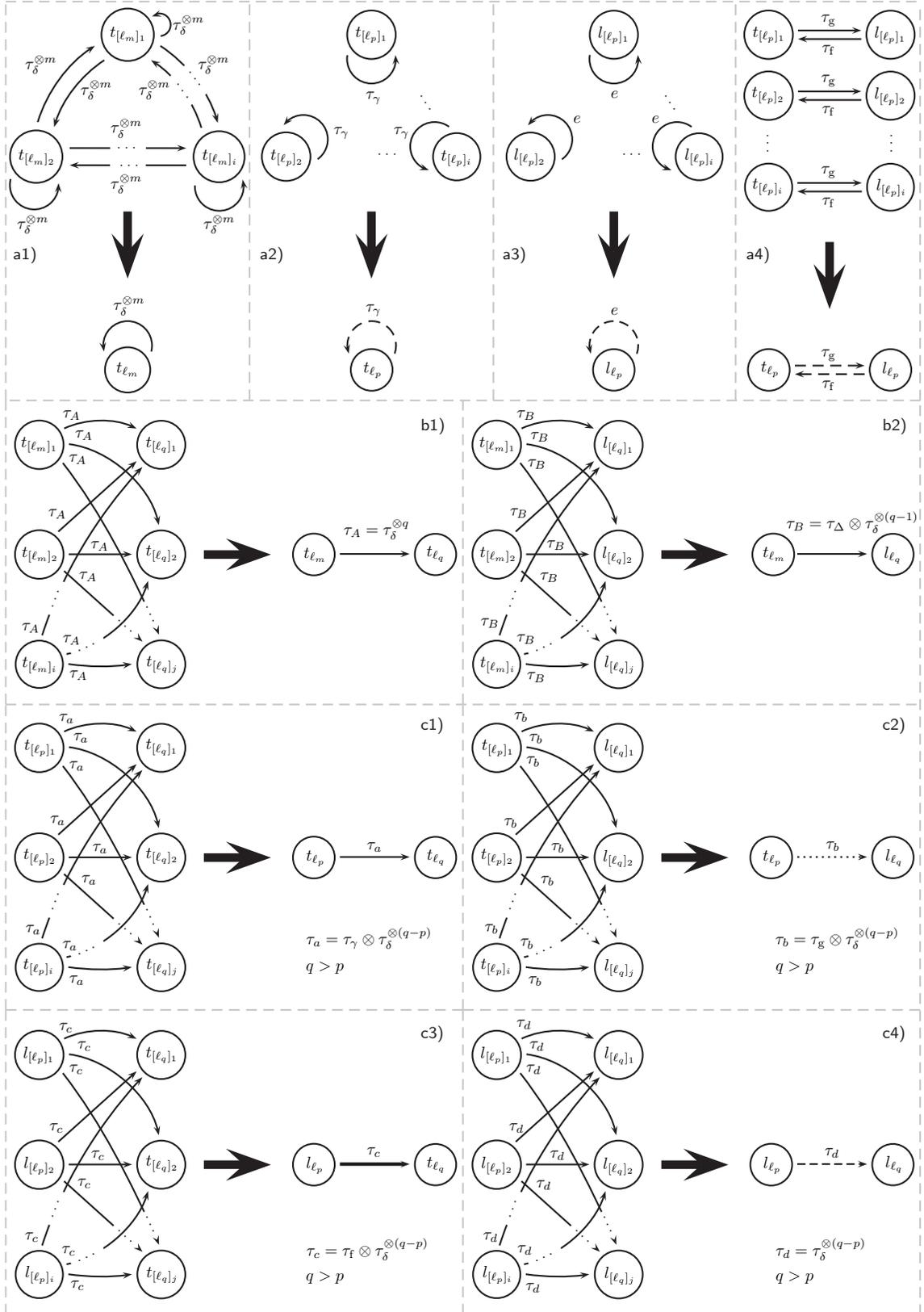
\figuracomprida{GraphReductions}{15.2cm}{Graph reductions. Touchdown and lift off events with indexes belonging to the same set $\leg{q}$ can be grouped together since they have the same number of output and input arcs with the same weights. }{}

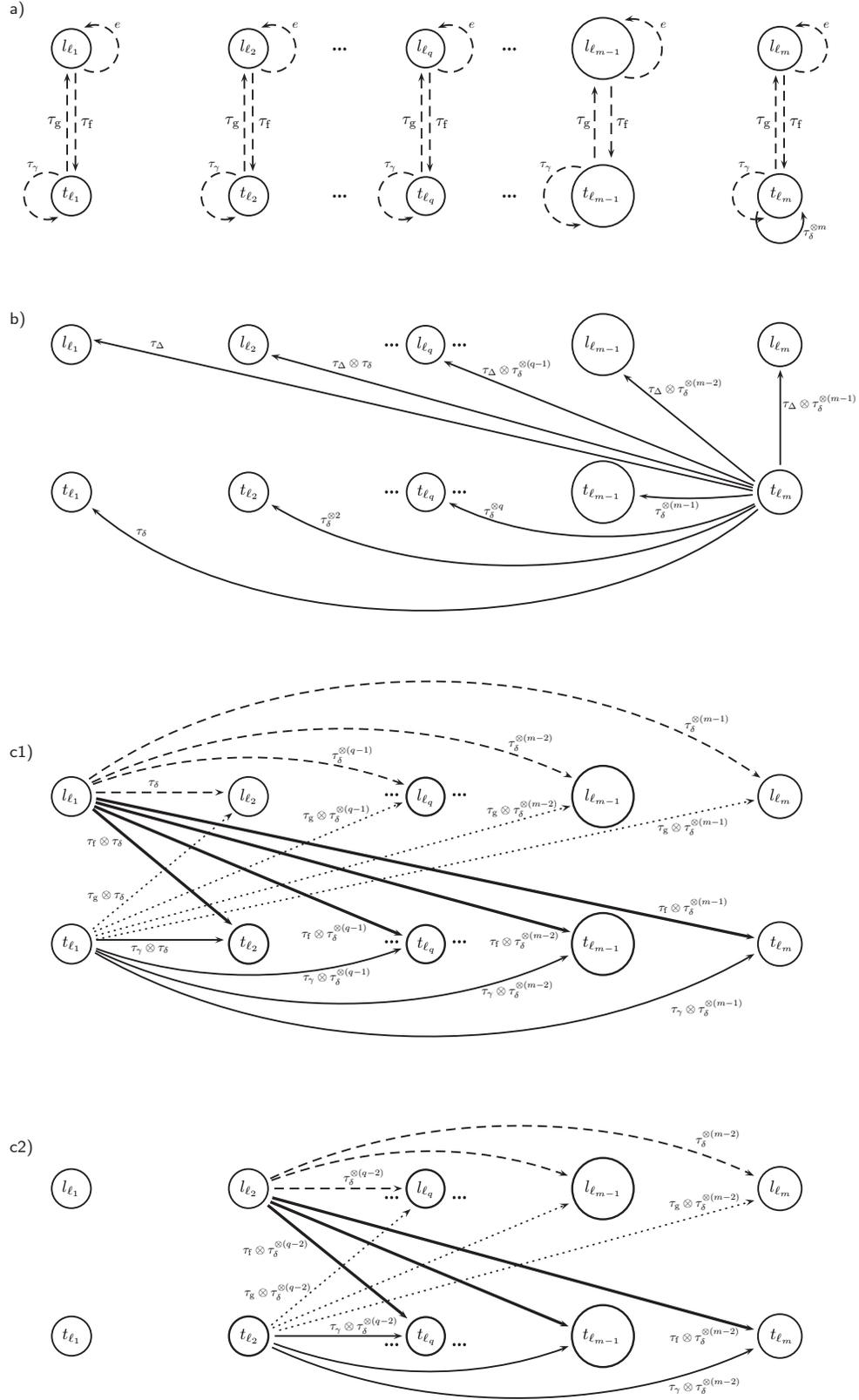
\figuracomprida{GraphConstruction}{13cm}{Elements of the precedence graph of the system matrix $A$. The total precedence graph of $A$ is composed of all the arcs presented in a) and b), together with the $m-1$ remaining subgraphs that follow the pattern of Figures c1) and c2).}{}

Consider the following assumption:\\[-2mm]
\begin{assumption}[A2] 
$\taugamma\leq  \Mpower{\Tz}{\m}$\\[-1mm]
\end{assumption}
\begin{lemma}\labellem{sscs}
If assumption A2 is verified then the critical graph of $\graphc(A)$ (and $\graphc(\Abar)$) has a single strongly connected subgraph.
\end{lemma}
\begin{proof}
We consider two cases:\\\\
i) $\taugamma =  \Mpower{\Tz}{\m}=\eg$.\\In this situation the circuits presented in Figures \ref{fg.GraphReductions}-a1 and \ref{fg.GraphReductions}-a2 all belong to the critical graph since their weights are $\taugamma$ or $\Mpower{\Tz}{\m}$ both equal to the max-plus eigenvalue $\eg$. Note that any circuit $c_{1}$ of length $l$ made from the nodes of $t_{\leg{m}}$, illustrated in \reffig{GraphReductions}-a1, has an average weight of
\NEn{\frac{|c_{1}|_{w}}{|c_{1}|_{1}}=\frac{\left(\Tdf^{\otimes m}\right)^{\otimes l}}{l}=\Tdf^{\otimes m}=\eg}
and as such also belongs to the critical graph.
\Gshrink
Any other circuit in the precedence graph of $\bar{A}$ must pass through at least one node of $t_{\leg{m}}$, as illustrated in Figures \ref{fg.GraphConstruction}-b,  \ref{fg.GraphConstruction}-c1, and  \ref{fg.GraphConstruction}-c2 (with the exception of the self-loops in \reffig{GraphReductions}-a3 and the circuits in \reffig{GraphReductions}-a4 that we don't consider since their weights are $\ze$ and $\taugamma/2$ both less then $\eg$). Additionally, arcs starting in nodes from a group $t_{\leg{q}}$ with $q<m$ are only connected to nodes in $t_{\leg{q+p}}$ for $p\geq 0$ (or $l_{\leg{q+p}}$). This is again illustrated in Figures \ref{fg.GraphConstruction}-a, \ref{fg.GraphConstruction}-c1, and  \ref{fg.GraphConstruction}-c2. Let $t_{[\leg{q}]_{i}}$ denote element $i$ of $t_{\leg{q}}$. Consider the circuit
\NEn{c_{2}:t_{[\leg{m}]_{i}} \to t_{[\leg{q}]_{j}}\to t_{[\leg{m}]_{i}}}
with $q<m$. The average weight is (with $\taugamma =  \Mpower{\Tz}{\m}$)
\NEn{\frac{|c_{2}|_{w}}{|c_{2}|_{1}}=\frac{\Tdf^{\otimes q}\otimes \taugamma\otimes \Tdf^{\otimes (m-q)}}{2}=\frac{\Tdf^{\otimes m}\otimes \taugamma}{2}=\eg}
Circuit $c_{2}$ is thus also in the critical graph. For the general circuit of the type
\NEn{c_{3}:t_{[\leg{m}]_{i}} \to \underbrace{ t_{[\leg{q_{1}}]_{j_{1}}}\to t_{[\leg{q_{2}}]_{j_{2}}}\to \cdots \to t_{[\leg{q_{l}}]_{j_{l}}}}_{l \text{~nodes}}\to t_{[\leg{m}]_{i}}}
with $q_{1}<q_{2}<\cdots<q_{l}<m$, the average weight is
\shrinkornot{
\NE{\frac{|c_{3}|_{w}}{|c_{3}|_{1}}\ele\frac{\taugamma^{\otimes l}\otimes \Tdf^{\otimes q_{1}}\otimes  \Tdf^{\otimes (q_{2}-q_{1})}\otimes \cdots \otimes  \Tdf^{\otimes (m-q_{l})}}{l+1}\\
\ele \frac{\taugamma^{\otimes l}\otimes\Tdf^{\otimes m}}{l+1}=\eg.}
}{
\NEn{\frac{|c_{3}|_{w}}{|c_{3}|_{1}}\ele\frac{\taugamma^{\otimes l}\otimes \Tdf^{\otimes q_{1}}\otimes  \Tdf^{\otimes (q_{2}-q_{1})}\otimes \cdots \otimes  \Tdf^{\otimes (m-q_{l})}}{l+1}= \frac{\taugamma^{\otimes l}\otimes\Tdf^{\otimes m}}{l+1}=\eg.}
}
Again, circuit $c_{3}$ is part of the critical graph. Any circuit that passes through any node in $l_{\leg{q}}$, for any $q$, will never be in the critical graph. This is due to the fact that arcs within touchdown nodes of different leg groups yield a higher weight:
\NEn{
t_{[\leg{q}]_{i}}\to t_{[\leg{p}]_{j}}&~~&\text{weight:}~\taugamma\otimes\Td^{\otimes (q-p)}\\
t_{[\leg{q}]_{i}}\to l_{[\leg{p}]_{j}}&~~&\text{weight:}~\Tg\otimes\Td^{\otimes (q-p)}\\
l_{[\leg{q}]_{i}}\to t_{[\leg{p}]_{j}}&~~&\text{weight:}~\Tf\otimes\Td^{\otimes (q-p)}\\
l_{[\leg{q}]_{i}}\to l_{[\leg{p}]_{j}}&~~&\text{weight:}~\Td^{\otimes (q-p)}
}
As such, a path that connects a touchdown node to a lift off node ``loses'' $\taugamma-\Tg=\Tf$ from the maximum possible weight, a path from lift off to lift off nodes loses $\taugamma$, and a path from lift off nodes to touchdown nodes loses $\Tg$ in weight.
This can also be observed in the structure of $\bar{A}$, in equation \refeq{verybigAexplicitA}, where the sub-matrix $\Tf\otimes(\Tg\otimes W\oplus V)$ overcomes the sub-matrices $\Tg\otimes W\oplus V$, $\Tf\otimes W$, and $W$.
Consider, for example, the circuit $c_{4}$:
\NEn{c_{4}:t_{[\leg{m}]_{i}} \to t_{[\leg{p}]_{j_{0}}}\to l_{[\leg{p+q}]_{j_{q}}} \to t_{[\leg{m}]_{i}}}
then
\shrinkornot{
\NE{\frac{|c_{4}|_{w}}{|c_{4}|_{1}}\ele\frac{\Tdf^{\otimes p}\otimes( \Tg\otimes\Tdf^{\otimes q})\otimes( \Tf\otimes \Tdf^{\otimes (m-(p+q))})}{3}\\
\ele\frac{\taugamma\otimes\Tdf^{\otimes m}}{3}<\eg.
}}
{
\NEn{\frac{|c_{4}|_{w}}{|c_{4}|_{1}}\ele\frac{\Tdf^{\otimes p}\otimes( \Tg\otimes\Tdf^{\otimes q})\otimes( \Tf\otimes \Tdf^{\otimes (m-(p+q))})}{3}=\frac{\taugamma\otimes\Tdf^{\otimes m}}{3}<\eg.
}}
Since all the nodes in the critical graph are connected (they are all touchdown nodes) we conclude that for the case $\taugamma=\Mpower{\Tz}{\m}=\eg$ the critical graph of $\bar{A}$ has a single strongly connected subgraph. \reffig{CriticalGraphs}-a illustrates the complete critical graph of $\bar{A}$ for this case.\\\\
ii) $\taugamma < \Mpower{\Tz}{\m}=\eg$.\\
In this situation only circuits of the type $c_{1}$ are part of the critical graph. Circuits of the type $c_{2}$ or $c_{3}$ are not part of the critical graph. \reffig{CriticalGraphs}-b illustrates the resulting critical graph of $\bar{A}$. Since all the nodes of $t_{\leg{m}}$ are connected to each other we conclude that for the case $\taugamma<\Mpower{\Tz}{\m}=\eg$ the critical graph of $\bar{A}$ has a single strongly connected subgraph. 

A third case can be considered: $\Mpower{\Tz}{\m}< \taugamma=\eg$. In this situation the critical graph of $\bar{A}$ does not have a single strongly connected subgraph. \reffig{CriticalGraphs}-c illustrates this situation, that we document here without proof.
\end{proof}

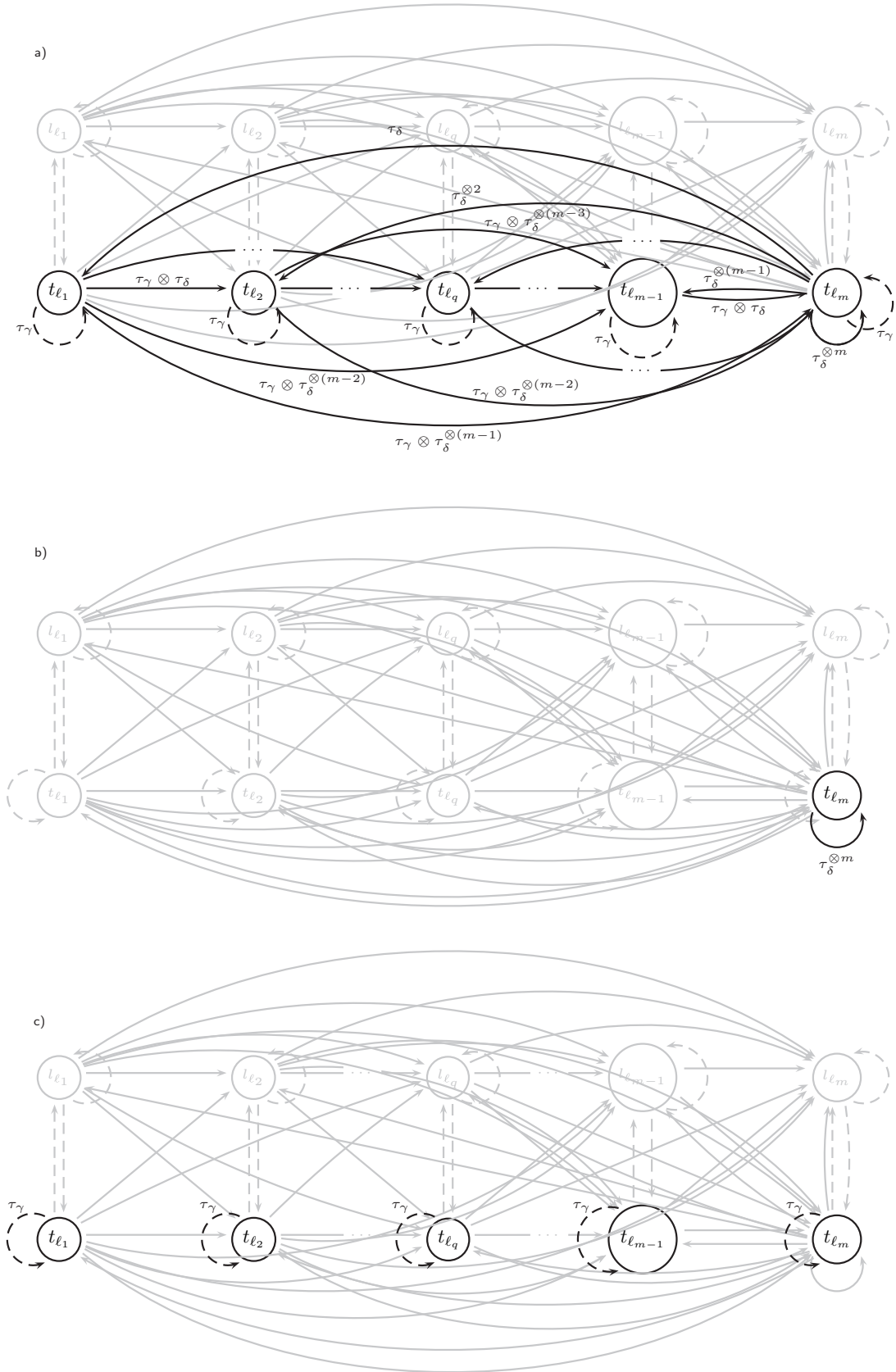
\figuracomprida{CriticalGraphs}{14.5cm}{Critical graphs of the system matrix $\bar{A}$. a) Case 1: $\taugamma = \Mpower{\Tz}{\m}=\eg$. b) Case 2: $\taugamma < \Mpower{\Tz}{\m}=\eg$. c) Case 3: $\Mpower{\Tz}{\m}< \taugamma=\eg$.
}{}

\begin{theorem}
Given assumptions A1 and A2, the max-plus eigenvalue $\eg$ of the system matrix $A$ (and $\Abar$), defined by equation \refeq{eigenvaluedefinition}, is unique, and the max-plus eigenvector $\ev$ of $A$ (and $\Abar$), defined by equations \refeq{eigenvectordefinitionA}--\refeq{eigenvectordefinitionB} is unique up to a max-plus scaling factor.
\end{theorem}
\begin{proof}
According to \reflem{irreducibility} $A$ is irreducible, and as such it has a unique max-plus eigenvalue. According to \reflem{sscs} the critical graph of $\graphc(A)$ has a single strongly connected subgraph, and as such its max-plus eigenvector is unique up to a max-plus scaling factor (see \cite{baccelli92}, Theorem 3.101).
\end{proof}

%
%

\subsection{Coupling time}\labelsec{coupling}

\refth{coupling} on page \pageref{th.coupling} describes an important property of max-plus-linear systems when the system matrix $A$ is irreducible: it guarantees the existence of an autonomous steady-state regime that is achieved in a number of finite steps $\couplingz$, called the \emph{coupling time}. Computing the coupling time is very important for the application of legged locomotion since it provides the number of steps a robot needs to take to reach steady state after a gait transition or a perturbation.
\begin{lemma}\labellem{couplingproof}
Given assumptions A1, A2, the coupling time for the max-plus-linear system defined by equation \refeq{sssync} is $\couplingz=2$ with cyclicity $c=1$. 
\end{lemma}
\begin{proof}
Computing successive products of $\An$ and taking advantage of its structure and equations \refeq{minusa}-\refeq{minus} one can write its \mbox{$\coupling$-th} power $\An^{\mtimes \coupling}$, valid for all $\coupling \geq 2$, illustrated by equations \refeq{ageneralA} and \refeq{ageneralB}. By inspection of the expression of $\An^{\mtimes \coupling}$ in \refeq{ageneralA}--\refeq{ageneralB} one can observe that most terms are max-plus multiplying by a power of the max-plus eigenvalue $\eg$ (recall that with assumption A2 we have $\eg= \Mpower{\Tz}{\m}$). To factor out $\eg$ of the matrix composed by expressions \refeq{ageneralA} and\refeq{ageneralB} we show that
\NEn{\nonumber
&&\eg^{\mtimes(\coupling-2)}\mtimes\Tf\mtimes\Tg\mtimes\Mv\mtimes\Md
\geq\Tf^{\mtimes(\coupling-1)}\mtimes\Tg^{\mtimes \coupling}\mtimes\Md
\Leftrightarrow\\
&&\Tg\mtimes\eg^{\mtimes(\coupling-2)}\mtimes\Tf\mtimes\Mv\mtimes\Md
\geq
\Tg\mtimes\taugamma^{\mtimes(\coupling-1)}\mtimes\Md
}
Since $\eg\geq \taugamma$ it is sufficient to show that
\NEn{
&&\Tf\mtimes\Mv\mtimes\Md
\geq
\taugamma\mtimes\Md
}
 This can be confirmed by inspecting equations \refeq{deltamatrix} and \refeq{vw}:
\begin{enumerate}
\item All the terms in the upper block triangle of $\taugamma\mtimes\Md$ are $\zz$ while for $\Tf\mtimes\Mv\mtimes\Md$ they are positive numbers.
\item In the block diagonal $[\Tf\mtimes\Mv\mtimes\Md]_{i,i}=\Tdf^{\otimes m}\otimes\Mone \geq \taugamma \otimes E=[\taugamma\mtimes\Md]_{i,i}$, by assumption A2. 
\item In the lower block triangle $[\Tf\mtimes\Mv\mtimes\Md]_{i,j}=\Tdf^{\otimes (m+i-j)}\otimes\Mone \geq \Tdf^{\otimes (i-j)}\otimes\taugamma \otimes \Mone=[\taugamma\mtimes\Md]_{i,j}$, by assumption A2. 
\end{enumerate}
 Taking advantage of this simplification one can obtain equations \refeq{ageneralAsimple}, \refeq{ageneralBsimple}, and  \refeq{aa}--\refeq{ab}. Together with the similarity transformation we obtain the result valid for $\coupling \geq 2$:
\shrinkornot{ 
\NE{ \A^{\mtimes (\coupling+1)}\ele \CC\mtimes\An^{\mtimes (\coupling+1)}\mtimes\CC^{T}\\
\ele\CC\mtimes\eg\mtimes \An^{\mtimes \coupling}\mtimes\CC^{T}=\eg\mtimes \A^{\mtimes \coupling},}}
{ 
\NEn{ \A^{\mtimes (\coupling+1)}\ele \CC\mtimes\An^{\mtimes (\coupling+1)}\mtimes\CC^{T}=\CC\mtimes\eg\mtimes \An^{\mtimes \coupling}\mtimes\CC^{T}=\eg\mtimes \A^{\mtimes \coupling},}
}
thus concluding that the coupling time is $\couplingz=2$ with cyclicity $c=1$.\end{proof}
\begin{figure*}[!t]
\shrinkornot{\normalsize}{\footnotesize}
\newcommand{\tttt}{\mtimes}
\newcommand{\pppp}{\mplus}
\NEn{\left[\An^{\mtimes \coupling}\right]_{\cdot,1}\ele\matris{cc}
{\Tf\tttt\left(
\eg^{\mtimes(\coupling-2)}\tttt\Tf\tttt\Tg\tttt\Mv\tttt\Md
\pppp
\eg^{\mtimes(\coupling-1)}\tttt\Mv
\pppp
\Tf^{\mtimes(\coupling-1)}\tttt\Tg^{\mtimes \coupling}\tttt\Md\right)\\
\eg^{\mtimes(\coupling-2)}\tttt\Tf\tttt\Tg\tttt\Mv\tttt\Md
\pppp
\eg^{\mtimes(\coupling-1)}\tttt\Mv
\pppp
\Tf^{\mtimes(\coupling-1)}\tttt\Tg^{\mtimes \coupling}\tttt\Md
}\labeleq{ageneralA}
\\
\ele \matris{cc}
{\Tf\tttt\left(
\eg^{\mtimes(\coupling-2)}\tttt\Tf\tttt\Tg\tttt\Mv\tttt\Md
\pppp
\eg^{\mtimes(\coupling-1)}\tttt\Mv
\right)\\
\eg^{\mtimes(\coupling-2)}\tttt\Tf\tttt\Tg\tttt\Mv\tttt\Md
\pppp
\eg^{\mtimes(\coupling-1)}\tttt\Mv
}\labeleq{ageneralAsimple}
}

\NEn{\left[\An^{\mtimes \coupling}\right]_{\cdot,2}\ele\matris{cc}
{\Tf\tttt\left(
\eg^{\mtimes(\coupling-2)}\tttt\Tf\tttt\Mv\tttt\Md
\pppp
(\Tf\tttt\Tg)^{\mtimes(\coupling-1)}\tttt\Md\right)\\
\eg^{\mtimes(\coupling-2)}\tttt\Tf\tttt\Mv\tttt\Md
\pppp
(\Tf\tttt\Tg)^{\mtimes(\coupling-1)}\tttt\Md
}\labeleq{ageneralB}
\\
\ele\matris{cc}
{\Tf\tttt\left(
\eg^{\mtimes(\coupling-2)}\tttt\Tf\tttt\Mv\tttt\Md
\right)\\
\eg^{\mtimes(\coupling-2)}\tttt\Tf\tttt\Mv\tttt\Md
}\labeleq{ageneralBsimple}
}

\NEn{\Tf\mtimes\Mv\mtimes\Md=\matris{ccc}
{
\Tdf^{\mtimes m}\mtimes\Monexb{1,1}&\cdots&\Tdf\mtimes\Monexb{1,m}\\
\vdots&\ddots&\vdots\\
\Tdf^{\mtimes (2m-1)}\mtimes\Monexb{m,1}&\cdots&\Tdf^{\mtimes m}\mtimes\Monexb{m,m}
}\labeleq{vw}
}

\newcommand{\ttttt}{\mtimes}
\newcommand{\ppppp}{\mplus}

\NEn{\labeleq{aa}\An^{\mtimes (\coupling +1)}\ele\matris{cc}
{\Tf\ttttt\left(
\eg^{\mtimes(\coupling-1)}\ttttt\Tf\ttttt\Tg\ttttt\Mv\ttttt\Md
\ppppp
\eg^{\mtimes \coupling}\ttttt\Mv\right)
&~~~~~~\Tf\ttttt\left(
\eg^{\mtimes(\coupling-1)}\ttttt\Tf\ttttt\Mv\ttttt\Md
\right)\\
\eg^{\mtimes(\coupling-1)}\ttttt\Tf\ttttt\Tg\ttttt\Mv\ttttt\Md
\ppppp
\eg^{\mtimes \coupling}\ttttt\Mv
&~~~~~~
\eg^{\mtimes(\coupling-1)}\ttttt\Tf\ttttt\Mv\ttttt\Md
}
\\
\ele\eg\mtimes\matris{cc}
{\Tf\ttttt\left(
\eg^{\mtimes(\coupling-2)}\ttttt\Tf\ttttt\Tg\ttttt\Mv\ttttt\Md 
\ppppp
\eg^{\mtimes(\coupling-1)}\ttttt\Mv\right)
&~~~~~~\Tf\ttttt\left(
\eg^{\mtimes(\coupling-2)}\ttttt\Tf\ttttt\Mv\ttttt\Md
\right)\\
\eg^{\mtimes(\coupling-2)}\ttttt\Tf\ttttt\Tg\ttttt\Mv\ttttt\Md
\ppppp
\eg^{\mtimes(\coupling-1)}\ttttt\Mv
&~~~~~~
\eg^{\mtimes(\coupling-2)}\ttttt\Tf\ttttt\Mv\ttttt\Md
}\\
\ele\eg\mtimes \An^{\mtimes \coupling}\labeleq{ab}
}

\hrulefill
\vspace*{4pt}

\end{figure*}

%
%

\section{Conclusions}

We have shown that max-plus linear models are very well suited to model leg synchronization in legged locomotion. The abstraction of the continuous-time dynamics associated with each leg as a two-event circuit simplifies the construction of a gait space. The combinatorial nature of this gait space, arising from all the possible arrangements in which multiple legs can be synchronized, is captured by a compact representation as an ordered set of ordered sets. 

For the classes of switching max-plus linear systems developed in this paper important structural properties of the max-plus system matrix are obtained in closed-form. The unique max-plus eigenvalue represents the total cycle time, the unique (up to a scaling factor) max-plus eigenvector dictates a unique steady-state behavior, and the coupling time reveals the transient response to gait switching or disturbances. 

Although some of the derivations in this paper are quite lengthy, using a combination of discrete-event algebraic tools with graph-theoretic concepts has lead to an important result in robotics: since the coupling time was found to be two, we have shown that legged robots can switch gaits (or rhythms) or recover from large disturbances in at least two steps. In the derivation process we have found that similarity transformations can facilitate the algebraic manipulations by exposing the structure of the system matrices. This was important to find closed-form expressions to the eigen-structure of the system and the coupling time, that typically need to be computed via simulations of using numerical procedures. On the graph side, the node-reduction procedure has allowed depicting graphs that can have an arbitrary large number of nodes. A graph-theoretic proof is needed for proving the uniqueness of the max-plus eigenvector. Our results are valid for robots with an arbitrary number of legs.

Further research will look towards relaxing the structure of the system matrix to address the synchronization of general cyclic systems and towards the modeling of more general gaits.

\bibliographystyle{plain} 

\ifthenelse{\isundefined{\shrinkpaper}}
{\bibliography{string-full,references}}
{\bibliography{string-abb,references}}

\end{document}